\documentclass[letterpaper, 10 pt, conference]{ieeeconf}  

\IEEEoverridecommandlockouts                              
\usepackage{graphicx}
\usepackage{color}
\usepackage{hyperref}
\usepackage{amsmath,amssymb,amsfonts,amsthm}
\usepackage{mathtools}
\usepackage{braket}
\usepackage{caption,subcaption}
\usepackage{url}
\usepackage{multirow}
\usepackage{algorithmic}
\usepackage[ruled,vlined]{algorithm2e}
\usepackage{tikz,tikz-cd}
\usetikzlibrary{positioning,arrows.meta,fit,calc}

\DeclareMathOperator{\rank}{rank}
\DeclareMathOperator{\im}{im}
\DeclareMathOperator{\eig}{eig}
\newcommand{\PAR}[1]{{\left( #1 \right)}}

\newcommand{\PPPAR}[1]{{\left[ #1 \right]}}
\newcommand{\R}{\mathbb{R}} 
\newcommand{\bbS}{\mathbb{S}} 
\newcommand{\PSD}{\mathbb{S}_{+}}
\newcommand{\PD}{\mathbb{S}_{++}} 
\newcommand{\ep}{{\varepsilon}}
\newcommand{\diag}[1]{{\rm diag}\PAR{#1}}
\newcommand{\norm}[1]{{\left\| {#1} \right\|}}

\newcommand{\mat}[1]{\begin{bmatrix}#1\end{bmatrix}}

\newcommand{\PROB}[1]{$\mathbf{P_{#1}}$}
\newcommand{\rev}[1]{{#1}}

\newtheorem{theorem}{Theorem}
\newtheorem{proposition}{Proposition}
\newtheorem{assumption}{Assumption}
\newtheorem{corollary}{Corollary}
\newtheorem{lemma}{Lemma}
\newtheorem{example}{Example}
\newtheorem{remark}{Remark}

\usepackage[backend=biber,style=ieee,url=false,doi=false,eprint=false, date=year]{biblatex}
\AtEveryBibitem{\clearfield{pages}}

\title{\LARGE \bf
Bilevel MPC for Linear Systems: A Tractable Reduction and Continuous Connection to Hierarchical MPC
}

\author{Ryuta Moriyasu$^{1,2}$, Carmen Amo Alonso$^{1}$ and Marco Pavone$^{1,3}$
\thanks{This manuscript has been accepted to the IEEE Conference on Decision and Control (CDC) 2026 for publication.
This work was supported by Toyota Motor Corporation}
\thanks{{\footnotesize$^{1}$ Department of Aeronautics \& Astronautics, Stanford University. \phantom{abcd}
        {\tt\footnotesize \phantom{ab} \{moriyasu,camoalon,pavone\}@stanford.edu}}}%
\thanks{{\footnotesize$^{2}$ Toyota Motor North America.
        {\tt\footnotesize ryuta.moriyasu@toyota.com}}}%
\thanks{{\footnotesize$^{3}$ NVIDIA Research.}}%
}

\begin{document}

\maketitle
\thispagestyle{empty}
\pagestyle{empty}

\begin{abstract} Model predictive control (MPC) has been widely used in many fields, often in hierarchical architectures that combine controllers and decision-making layers at different levels. 
However, when such architectures are cast as bilevel optimization problems, standard KKT-based reformulations often introduce nonconvex and potentially nonsmooth structures that are undesirable for real-time verifiable control.
In this paper, we study a bilevel MPC architecture composed of (i) an upper layer that selects the reference sequence and (ii) a lower-level linear MPC that tracks such reference sequence. 
We propose a smooth single-level reduction that does not degrade performance under a verifiable block-matrix nonsingularity condition.
In addition, when the problem is convex, its solution is unique and equivalent to a corresponding centralized MPC, enabling the inheritance of closed-loop properties. 
We further show that bilevel MPC is a natural extension of standard hierarchical MPC, and introduce an interpolation framework that continuously connects the two via move-blocking. This framework reveals optimal-value ordering among the resulting formulations and provides inexpensive \emph{a posteriori} degradation certificates, thereby enabling a principled performance–computational efficiency trade-off.
Code for this project is available at \url{https://github.com/StanfordASL/Reduced_BMPC}.

\end{abstract}

\section{Introduction}

\subsection{Background}

    Hierarchical control architectures, where upper-level decisions (e.g., economics, safety, system-wide performance) are coupled with the optimal response of a lower-level controller, are widespread in process control~\cite{QinBadgwell2003,Scattolini2009}, automotive systems~\cite{falcone2008hierarchical,Garone2017}, power/microgrids~\cite{Guerrero2011_TIE_Microgrids}, buildings/HVAC~\cite{Serale2018_Energies_BuildingsMPC}, spacecraft~\cite{KolmanovskyBilevelCG2024}, traffic networks~\cite{Zhou2017_TCST_TrafficHMPC}, and others.
    One reason is that separating design by layer, where each has distinct functional-safety/responsibility requirements, is common practice~\cite{QinBadgwell2003,Benveniste2018Contracts}.
    Another reason is that confining learning or policy tuning to the upper layer is desirable~\cite{Hewing2020} as it preserves the structure of the lower-level controller and thereby makes safety, verification, and deployment more manageable.
    Hence, hierarchical paradigms are naturally preferred in many fields~\cite{Scattolini2009}.
    
    In the context of model predictive control (MPC), hierarchical MPC (HMPC)~\cite{Scattolini2009,Marchetti2014RTO_MPC} is widespread. 
    In this setting, the upper level selects steady-state targets/schedules and the lower level tracks them via dynamic MPC.\footnotemark
    Hence, HMPC is effective when the upper-level objective depends primarily on steady-state quantities and transient dynamics have limited influence~\cite{QinBadgwell2003,RawlingsMayneDiehl2017}.
    Yet, when transients affect cost or constraint usage, ignoring dynamics at the upper layer can incur non-negligible performance loss~\cite{Angeli2012}.
    To avoid such losses while retaining the hierarchical architecture, approaches for \emph{bilevel MPC} have been proposed~\cite{Avraamidou2019,8431884}.
    
    Bilevel MPC is formulated as a \emph{bilevel optimization} problem in which the upper-level optimization is constrained by the solution of a lower-level optimization problem~\cite{Colson2007,Dempe2002,Bard1998}.
    A standard approach, especially when the lower-level problem is convex and regular, is to replace the lower-level problem with its KKT conditions to obtain a single-level \emph{mathematical program with complementarity constraints} (MPCC) (see, e.g.,~\cite{luo1996mathematical}).
    In control, this approach is used in bilevel command governors~\cite{KolmanovskyBilevelCG2024}, RTO--MPC embedding of KKT systems~\cite{Marchetti2014RTO_MPC}, differentiable MPC layers~\cite{Amos2018}, and safety-oriented hierarchies~\cite{kogel2025safe}, among others.
    \footnotetext{The term ``hierarchical MPC'' is sometimes used with different meanings in the literature.
    Here, HMPC refers to the two-layer architecture with an upper-level steady-state target optimizer (or scheduler)
    and a lower-level tracking MPC, as formalized in Section~IV.}
    
    Despite the widespread interest in bilevel MPC across different applications, current approaches still face significant challenges.
    In particular, MPCC formulations are inherently nonconvex due to complementarity and possible constraint qualification (CQ) violations (LICQ/MFCQ)~\cite{ScheelScholtes2000_MP,luo1996mathematical,ye2006constraint}, which can lead to local/multiple solutions and numerical instabilities.
    A common workaround is to encode complementarity with binaries, yielding mixed-integer formulations. 
    However, such reformulations lack scalability and cause policy discontinuities, which often leads to nonuniqueness and switching~\cite{Kleinert2021Survey,Kleinert2020BigM,MOSEKCookbook2023}.
    These drawbacks are at odds with control applications where speed, continuity, and uniqueness are critical.
    
\subsection{Aims and Contributions}

    Motivated by the issues with MPCC, we seek a reduced bilevel MPC formulation that is solvable as a convex program to ensure the uniqueness of the solution while maintaining optimality.
    Particularly, we assume a bilevel MPC architecture where an upper layer selects a reference sequence that is given to the lower-level linear tracking MPC formulated as a strongly convex quadratic program (QP).
    Focusing on settings where the upper layer does not directly depend on the reference sequence, we provide a tractable alternative to MPCC-based bilevel MPC.
    Moreover, to make the resulting bilevel design easily deployable and comparable to standard HMPC baselines, a move-blocking~\cite{cagienard2007move} technique, which divides the horizon into some blocks and fixes the signal in each block, is introduced to connect bilevel MPC to HMPC and enable a principled trade-off between computational efficiency and performance.  
    
    \begin{itemize}[\labelindent=0pt]
      \item \textbf{Tractable reduction (QP-solvable bilevel consistency).}
      We propose a smooth, complementarity-free single-level formulation by lifting the lower-level inequalities to the upper level and enforcing only stationarity.
      Under a verifiable block-matrix nonsingularity condition, it is exact: it recovers the same optimal value and realized input trajectory as the original bilevel MPC, independently of the active-set pattern.
      Moreover, under a standard convexity assumption, the realized input optimizer is unique and coincides with that of a corresponding centralized MPC.
      \item \textbf{HMPC--BMPC connection (performance--complexity trade-off with certificates).}
      We interpret standard HMPC as a one-block reference policy and connect it to bilevel MPC via move-blocking~\cite{cagienard2007move}.
      We establish optimal-value ordering among hierarchical, bilevel, and centralized formulations and their blocked variants, and derive inexpensive a posteriori certificates that upper-bound the degradation induced by blocking.
      We also provide a constructive procedure to build low-rank blocking matrices from HMPC data that improve over HMPC on sampled points/regions.
    \end{itemize}
    In this paper, some proofs and remarks are omitted and provided in the appendix of the arXiv preprint \cite{Moriyasu2026bmpc}.

    \subsubsection*{Notation}
    Let us denote 
    the set of all $n\times n$ symmetric matrices by $\bbS^n$ and symmetric positive (semi-)definite matrices by $\bbS^n_{++} \, (\bbS^n_{+})$;
    the identity matrix and zero matrix of appropriate size by $I$ and $O$, respectively;
    the identity matrix with an explicit size of $n\times n$ by $I_n$;
    the $n$-dimensional vector whose all elements are $1$ by $\mathbf{1}_n$;
    the Moore–Penrose inverse of a matrix $A$ by $A^\dagger$;
    the $i$-th element of a vector $v$ by $(v)_i$; 
    the sub-vector consisting of $(v)_i, i \in J$ by $(v)_J$; 
    a collective vector $[v_1^\top \cdots v_n^\top]^\top$ by $[v_1;\cdots;v_n]$; 
    Kronecker product by $\otimes$; 
    a weighted norm $(x^\top P x)^{1/2}$ with a vector $x\in\R^n$ and a matrix $P\in \PD^n$ by $\norm{x}_P$.
    
\section{Standard Bilevel MPC Formulation}\label{sec:problem}

In this section, we define the baseline bilevel MPC. This is a bilevel controller where the upper level selects the reference sequence as a stack of steady-state pairs, and the lower level decides the input to track the reference (see Fig.~\ref{fig:bilevel_structure}).
Therefore, we introduce a minimal parameterization of the steady-state pair and then formulate the bilevel MPC.

    \subsection{Plant and Steady-State}
    
    Consider a discrete-time LTI system
    \begin{equation*}
      x_{k+1} = A x_k + B u_k,
    \end{equation*}
    with state $x_k \in \mathbb{R}^{n}$ and input $u_k \in \mathbb{R}^{m}$.
    
    \begin{assumption}\label{ass:plant}
    The matrix $S:= [\, I-A,\, -B\,] \in \mathbb{R}^{n\times(n+m)}$ has full row rank $(\rank (S) = n)$.
    \end{assumption}\noindent
    The controllability of $(A,B)$ is sufficient to satisfy the above assumption, 
    but it is not a necessary condition.
    
    The set of all steady-state pairs $(\bar x, \bar u)$, i.e., $ S [\bar x; \bar u]=0$, is an $m$-dimensional linear subspace of $\R^{n+m}$ because $\dim\ker S = (n+m)- \rank(S) =  m$.
    We fix a basis of $\ker S$ and write the basis matrix as
    \[
      Z:= \mat{ Z_x\\ Z_u } \in \mathbb{R}^{(n+m)\times m} \ (\im Z = \ker S).
    \]
    Hence, any steady-state pair $(\bar x,\bar u)$ can be written as $[\bar x; \bar u] = Z\theta$ with a minimal-dimensional parameter $\theta\in\mathbb{R}^{m}$.
    
    \begin{remark}[Steady-state map]\label{rem:ssmap}
    If $I-A$ is nonsingular, one admissible choice for $Z$ is $Z_x = (I-A)^{-1}B$ and $Z_u = I_{m}$, in which case $\bar u = \theta, \, \bar x = (I-A)^{-1}B\,\bar u$.
    If not, i.e., for systems with integrators, one can use \texttt{Z = null(S)} (in MATLAB) to obtain one of the admissible $Z$.
    When $I-A$ is nonsingular but ill-conditioned, a safe choice is to use \texttt{Z = null(S, tol)} with numerical tolerance \texttt{tol}.
    \end{remark}

\begin{figure}[t]
\centering
\begin{tikzpicture}[
    font=\footnotesize,
    auto,
    node distance=1.7cm,
    >=latex,
    block/.style={
        draw,
        rectangle,
        minimum height=0.7cm,
        minimum width=1.0cm,
        align=center
    },
    line/.style={->, thick}
]

\node[block] (upper) {\PROB{1}};
\node[block, right=of upper, xshift=3mm] (lower) {\PROB{L}};
\node[block, right=of lower] (plant) {LTI};

\node[above=1pt of upper] {Upper MPC};
\node[above=1pt of lower] {Lower MPC};
\node[above=1pt of plant] {Plant};

\draw[line] (upper) -- node[above,align=center] {Reference\\sequence} node[below]{$\Theta$} (lower);
\draw[line] (lower) -- node[above] {Input} node[below]{$u$} (plant);
\draw[line] (plant) -- ++(2.0,0) node[midway, above,align=center] {State,\\Output} node[midway, below] {$x,y$};

\end{tikzpicture}
\caption{Bilevel MPC architecture. The upper MPC selects the reference sequence $\Theta$, the lower MPC computes the input $u$ to track $\Theta$, and the plant returns the state and output $(x,y)$.}
\label{fig:bilevel_structure}
\end{figure}
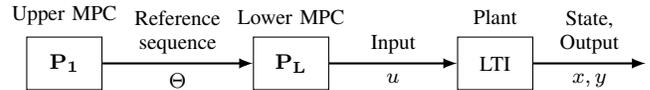

    \subsection{Control Problem}

    Collect the input sequence $U := [u_0; \dots ; u_{N-1}]$ over the horizon length $N$, and define the predicted state stack
    \begin{align*}
    & X(U,x_0) := [ x_1 ; \dots; x_N ] = \bar A x_0 + \bar B U,\\
    & \bar A   := \mat{ A\\ A^2\\ \vdots\\ A^N },\
      \bar B   := \mat{ B & O & \cdots & O\\
                       AB & B & \ddots & \vdots\\
                      \vdots & \ddots & \ddots & O\\
                      A^{N-1}B & \cdots & AB & B}.
    \end{align*}
    In addition, define the steady-state parameter sequence $\Theta := [\theta_0 ; \dots ; \theta_{N-1}]$, meaning that each stage-wise reference $(\bar x,\bar u)$ is a steady-state pair parameterized via $\theta$.
    Accordingly, we define the stacked steady-state references $\bar X(\Theta):=[\bar x_1;\dots;\bar x_N]$ and $\bar U(\Theta):=[\bar u_0;\dots;\bar u_{N-1}]$, with $[\bar x_{k+1};\bar u_k]=Z\theta_k$ for $k=0,\dots,N-1$.
    
    The lower-level MPC is defined as a strongly convex QP
    \begin{subequations}\label{eq:lower}
    \begin{align}
      & \mathbf{P_L}: \min_{U} \ F_l(U;\Theta,x_0) \ \text{s.t.}\ G_l(U;x_0) \le 0, \text{ where}\\
      & F_l := \tfrac12 \norm{ X(U,x_0) - \bar X(\Theta) }_{\bar Q}^2
      + \tfrac12 \norm{ U - \bar U(\Theta) }_{\bar R}^2, \label{eq:lowerobj}
    \end{align}        
    \end{subequations}
    $\bar Q\in \PSD^{Nn} $, $\bar R \in \PD^{Nm}$, and $G_l$ encode linear constraints (e.g., input, state, and output bounds).
    
    \begin{example}[Standard block?diagonal weights]\label{ex:example1}
    A standard weight setting for \PROB{L} is to use block?diagonal stacks:
    \begin{align*}
        \bar Q &= \mathrm{blkdiag}(\underbrace{Q,\dots,Q}_{N-1\ \text{times}},\,P),\ 
        \bar R = \mathrm{blkdiag}(\underbrace{R,\dots,R}_{N\ \text{times}}),
    \end{align*}
    with stage, terminal, and input weights $Q \in \PSD^{n}, \, P \in \PSD^n, \,\allowbreak R \in \PD^n$, respectively.
    \end{example}

    We now state the standard bilevel MPC studied in this paper.
    It chooses $\Theta$ to minimize an application-level objective $F_u(U;x_0)$ under upper-level constraints $G_u(U;x_0)\le 0$:
    \begin{subequations}\label{eq:P1}
    \begin{align}
    \mathbf{P_1}: \ & \min_{\Theta,\,U} \ F_u(U;x_0) \ \text{s.t.} \ G_u(U;x_0)\le 0,  \\
    & U = \arg\min_{V} \Set{ F_l(V;\Theta,x_0) | G_l(V;x_0)\le 0 },\label{eq:PLinP1}
    \end{align}
    \end{subequations}
    while enforcing that $U$ is generated by \PROB{L}.
    The standard method to solve \PROB{1} is to reformulate it as an MPCC:
    \begin{subequations}\label{eq:MPCC}
    \begin{align}
        \!\!\smash[b]{\min_{\Theta,U,\mu} F_u(U;x_0)} \, \text{s.t.} \,
        & G_u(U;x_0)\!\le\! 0, \nabla_U F_l \!+\! \nabla_U G_l \mu \!=\! 0, \!\! \\
        & G_l(U;x_0)\!\le\! 0, \ \mu \ge 0, \ G_l^\top \mu = 0,  \label{eq:cmpl}
    \end{align}       
    \end{subequations}
    where $\mu$ is the Lagrange multiplier.
    This problem is nonconvex due to its complementarity condition \eqref{eq:cmpl} even if $G_l$ is affine. It often causes issues such as a lack of reproducibility, insufficient performance, and discontinuity of the solution. 
    \rev{In this paper, $\min$ is treated as the global minimum.}

    \begin{remark}[Why we parameterize references via $\theta$]
        In tracking MPC, it is natural to restrict reference candidates to steady-state pairs.
        The parameterization
        (i) excludes unrealizable $(\bar X,\bar U)$, and thus, makes it easy to discuss stability,
        (ii) uses the minimal number of decision variables compared to treating $(\bar X,\bar U)$ with the constraint $(I-A)\bar X - B\bar U = 0$ explicitly, and also in our case,
        (iii) integrates smoothly with our method introduced in Section~\ref{sec:approx}.
    \end{remark}

\section{A Tractable Reduction}\label{sec:approx}

    In this section, we propose a smooth single-level reduction of \PROB{1} to avoid the issues with the complementarity condition \eqref{eq:cmpl} described above, and we show that the reduction can preserve performance and ensure the uniqueness of the solution and closed-loop stability under certain conditions.
    
    We propose the reduction for \PROB{1} by lifting the lower-level constraints to the upper level and replacing the lower-level problem with its \emph{unconstrained} first-order optimality condition. Hence, the complementarity condition disappears:\noindent
    \begin{subequations}\label{eq:P2}
    \begin{align}
        \!\! \mathbf{P_2}: \smash[b]{\min_{\Theta,\,U} F_u(U;x_0)} \ 
                \text{s.t.} \ & G_u(U;x_0) \! \le\! 0, G_l(U;x_0) \!\le\! 0, \\
                              & \nabla_U F_l(U;\Theta,x_0) = 0. \label{eq:P2grad}
    \end{align}
    \end{subequations}
    This can be viewed as a restriction of the design space of the MPCC \eqref{eq:MPCC} to $\mu=0$.
    The condition \eqref{eq:P2grad} is an affine equation in $U, \Theta$ for any fixed $x_0$.
    Indeed, differentiating \eqref{eq:lowerobj} yields
    \begin{align}
       \nabla_U F_l & = \bigl( \bar B^\top \bar Q\bar B + \bar R \bigr) U - \bar \Gamma \Theta +  \bar B^\top \bar Q \bar A x_0 = 0, \label{eq:grad} \\
       \bar \Gamma & := \bar B^\top \bar Q \bar Z_x + \bar R \bar Z_u \in\mathbb{R}^{N m\times N m}. \nonumber 
    \end{align}

    We consider an assumption on the matrix $\bar \Gamma$:
    \begin{assumption}\label{ass:nonsing}
        The matrix $\bar \Gamma$ is nonsingular.
    \end{assumption}\noindent
    \rev{This assumption is fairly mild. 
    Indeed, Appendix~I~\cite{Moriyasu2026bmpc} shows that (i) if the lower-level weights are tunable, it can be met by their design; and (ii) even if they cannot be changed, it holds in \emph{almost all} cases.}
    
    Crucially, under this assumption, the \emph{solution set}, which is the set of points that have the minimum value within the feasible set, of \PROB{2} becomes a subset of that of \PROB{1}.
    This implies that the optimal value does not change.
    \begin{theorem}[Feasible/solution-set inclusion]\label{thm:subset}
        Under Assumption~\ref{ass:plant} and \ref{ass:nonsing}, \rev{for $i=1,2$,} let $\mathcal F_i,\, V_i, \, \mathcal S_i$ be the feasible set, optimal value, and solution set of $\mathbf{P}_i$, for a fixed $x_0$.
        Then, 
        \[
        \text{\rm (1)} \ \mathcal F_2 \subseteq \mathcal F_1, \ 
        \text{\rm (2)} \ V_1 = V_2, \ 
        \text{\rm (3)} \ \mathcal S_2 \subseteq \mathcal S_1.
        \]
    \end{theorem}    
    \begin{proof}
        By Assumption~\ref{ass:nonsing}, for any $U$ and $x_0$, \eqref{eq:grad} admits the unique solution $\Theta = \Theta^\star(U;x_0)$, where
        \begin{align}\label{eq:Thetastar}
          \Theta^\star(U;x_0) := \bar \Gamma^{-1}
          \bigl( \bigl( \bar B^\top \bar Q\bar B + \bar R \bigr) U 
          + \bar B^\top \bar Q \bar A x_0 \bigr),
        \end{align} 
        which is obtained by solving \eqref{eq:grad} with respect to $\Theta$.\vspace{3pt}\\
        1) Let $(\Theta,U)$ be feasible for \PROB{2}. Then $G_l(U;x_0)\le0$ holds and
        $\Theta$ satisfies the lower-level stationarity. Hence $U$ solves the lower-level problem at $(\Theta,x_0)$, so $(\Theta,U)$ is feasible for \PROB{1}. This proves $\mathcal F_2\subseteq \mathcal F_1$.\vspace{3pt}\\
        2) From (1), $V_1\le V_2$. Conversely, take any feasible $(\Theta^\circ,U^\circ)$ of \PROB{1}. Replacing $\Theta^\circ$ by $\Theta^\star(U^\circ,x_0)$ preserves feasibility in \PROB{2} and leaves the upper-level objective unchanged as $F_u$ depends only on $U$. 
        Hence $V_2\le V_1$, which yields $V_1=V_2$.\vspace{3pt}\\        
        3) Let $(\Theta^\star,U^\star)\in\mathcal S_2$. By (1), it is feasible for \PROB{1}, and by (2), it attains the value $V_1$; therefore $(\Theta^\star,U^\star)\in\mathcal S_1$.
    \end{proof}\noindent
    This shows that enforcing the first-order stationarity relation in \PROB{2} provides a direct handle on bilevel-feasible pairs, i.e., one can obtain (a subset of) \PROB{1} solutions without tackling an MPCC \eqref{eq:MPCC} directly.
    In this sense, \PROB{2} provides a tractable route to solutions of \PROB{1}: solving \PROB{2} already returns candidates that are bilevel-consistent and hence eligible to be optimal for \PROB{1} once the upper objective $F_u$ is minimized.

    On the other hand, the whole bilevel MPC system, i.e., a cascade connection of the upper-level MPC (\PROB{1} or \PROB{2}), the lower-level MPC (\PROB{L}), and the plant, is still structurally difficult to analyze for its closed-loop properties, such as stability.
    Therefore, as an analytical tool for the bilevel system, we next introduce an associated centralized MPC:
    \begin{align}\label{eq:P3}
        \! \mathbf{P_3}: \, \min_{U} F_u(U;x_0) \ \text{s.t.} \ G_u(U;x_0) \! \le \!0, G_l(U;x_0) \! \le \! 0.\!
    \end{align}
    Under strong convexity, we show that the cascade of \PROB{1} or \PROB{2} and \PROB{L} is equivalent to the centralized MPC \PROB{3}, and therefore, the closed-loop properties are inherited from \PROB{3}.
    
    \begin{assumption}[Upper-level convexity]\label{ass:upper}
        For fixed $x_0$, $F_u(\cdot;x_0)$ and $G_u(\cdot;x_0)$ in \eqref{eq:P1} are strongly convex and convex, respectively, and $\mathcal C(x_0):=\{U\mid G_u(U;x_0)\le 0,\ G_l(U;x_0)\le 0\}\neq\emptyset$.
    \end{assumption}

    \begin{lemma}[Uniqueness of solution]\label{lem:uniqueness}
        Under Assumptions~\ref{ass:plant}, \ref{ass:nonsing}, and \ref{ass:upper}, for any fixed $x_0$,
        \PROB{1}, \PROB{2}, and \PROB{3} have the same unique optimizer in $U$, denoted by $U^\star$.
        In addition, \PROB{2} has the unique optimizer $(\Theta^\star(U^\star;x_0),U^\star)$.
    \end{lemma}
    \begin{proof} 
        Under Assumption~\ref{ass:nonsing}, using $\Theta^\star(U;x_0)$ in \eqref{eq:Thetastar}, the feasible set of \PROB{2} is $\mathcal F_2 = \{\,(\Theta^\star(U;x_0),U)\mid U\in\mathcal C(x_0)\,\}$, so its $U$-projection is exactly $\mathcal C(x_0)$. By Theorem~\ref{thm:subset}, $\mathcal F_2\subseteq\mathcal F_1$, and thus every $U\in\mathcal C(x_0)$ is attained by some feasible pair of \PROB{1}.
        Conversely, if $(\Theta,U)$ is feasible for \PROB{1}$,$ then necessarily $G_u(U;x_0)\le 0$ and $G_l(U;x_0)\le 0$, i.e., $U\in\mathcal C(x_0)$.
        Therefore, the $U$-projection of the feasible set of \PROB{1} is also exactly $\mathcal C(x_0)$.
        
        Thus, \PROB{1}, \PROB{2}, and \PROB{3} all reduce, at the level of $U$, to minimizing $F_u(U;x_0)$ over the same set $\mathcal C(x_0)$.
        Since $F_u(\cdot;x_0)$ is strongly convex by Assumption~\ref{ass:upper}, this minimization admits a unique optimizer $U^\star$.
        Hence all three problems share the same unique optimizer in $U$.
        Finally, since \PROB{2} associates a unique $\Theta^\star(U;x_0)$ with each $U$, its optimizer is uniquely given by $(\Theta^\star(U^\star;x_0),U^\star)$.
    \end{proof}

    \begin{theorem}[Equivalence to centralized MPC]\label{thm:equiv}
        Under Assumptions~\ref{ass:plant}, \ref{ass:nonsing}, and \ref{ass:upper}, for any fixed $x_0$,
        let $(\Theta, U)$ be any optimizer of \PROB{1} or \PROB{2}.
        Then, solving \PROB{L} with $\Theta$ returns the unique optimizer of \PROB{3}, namely $U^\star$.
        Hence, the cascades \PROB{1}$\to$\PROB{L} and \PROB{2}$\to$\PROB{L} are equivalent to \PROB{3} in their realized input (Fig.~\ref{fig:bilevel_structure_equiv}).
    \end{theorem}
    
    \begin{proof} 
        From Lemma~\ref{lem:uniqueness}, $U = U^\star$.
        For \PROB{1}, feasibility already means that $U^\star$ solves \PROB{L} with $\Theta$.
        For \PROB{2}, the conditions
        \[
        G_\ell(U^\star;x_0)\le 0,\
        \nabla_U F_\ell(U^\star;\Theta,x_0)=0
        \]
        are precisely the first-order optimality conditions of \PROB{L}.
        Since \PROB{L} is strongly convex, $U^\star$ is the unique global optimizer.
        Thus, in either case, solving \PROB{L} with $\Theta$ yields $U = U^\star$, which is also the unique optimizer of \PROB{3}.
    \end{proof}
    \noindent

    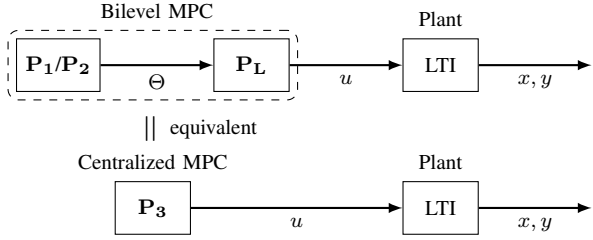
\begin{figure}[t]
    \centering
    \begin{tikzpicture}[
        font=\footnotesize,
        auto,
        node distance=1.5cm,
        >=latex,
        block/.style={
            draw,
            rectangle,
            minimum height=0.7cm,
            minimum width=1.0cm,
            align=center
        },
        line/.style={->, thick}
    ]
    
    \node[block] (upper) {\PROB{1}/\PROB{2}};
    \node[block, right=of upper] (lower) {\PROB{L}};
    \node[block, right=of lower] (plant1) {LTI};
    
    \node[above=1pt of plant1] {Plant};
    
    \coordinate (out1) at ($(plant1.east)+(1.5,0)$);
    
    \draw[line] (upper) --
        node[below] {$\Theta$}
        node[above=14pt] {Bilevel MPC}
        (lower);
    
    \draw[line] (lower) --
        node[below] {$u$}
        (plant1);
    
    \draw[line] (plant1) --
        node[midway, below] {$x,y$}
        (out1);
    
    \coordinate (fitup) at ($(upper.east)!0.5!(lower.west)+(0,0.32)$);
    \coordinate (fitdown) at ($(upper.east)!0.5!(lower.west)+(0,-0.32)$);
    
    \node[
        draw,
        dashed,
        rounded corners,
        inner sep=3pt,
        fit=(upper)(lower)(fitup)(fitdown)
    ] (group) {};
    
    \node[block] (p3) at ($(group.center)+(0,-1.9)$) {\PROB{3}};
    \node[block] (plant2) at ($(plant1.center)+(0,-1.9)$) {LTI};
    
    \node[above=1pt of p3] {Centralized MPC};
    \node[above=1pt of plant2] {Plant};
    
    \coordinate (out2) at ($(out1)+(0,-1.9)$);
    
    \draw[line] (p3) --
        node[below] {$u$}
        (plant2);
    
    \draw[line] (plant2) --
        node[midway, below] {$x,y$}
        (out2);
    
    \coordinate (eqpos) at ($(group.south)!0.35!(p3.north)$);
    
    \node[font=\Large, rotate=90] (eqmark) at (eqpos) {$=$};
    \node[font=\footnotesize, right=3pt of eqmark, yshift=-9pt] {equivalent};
    
    \end{tikzpicture}
    \caption{Equivalence between bilevel MPC and centralized MPC: the cascade of \PROB{1} or \PROB{2} and \PROB{L} realizes the same input as \PROB{3}, under Assumptions~\ref{ass:plant}, \ref{ass:nonsing}, and \ref{ass:upper}.}
    \label{fig:bilevel_structure_equiv}
    \end{figure}

    \begin{corollary}[Inheritance of closed-loop properties]\label{cor:stability}
    Under Assumptions~\ref{ass:plant}, \ref{ass:nonsing}, and \ref{ass:upper}, suppose the closed-loop system with \PROB{3} is (globally or locally) stable with respect to a fixed point $x_\infty$. 
    Then, the bilevel architectures combining \PROB{L} with \PROB{1} or \PROB{2} inherit the corresponding stability property.
    \end{corollary}
    \begin{proof} The proof follows directly from the equivalence between \PROB{3} and the cascade connection of \PROB{1} or \PROB{2} and \PROB{L}. The details are provided in Appendix~II~\cite{Moriyasu2026bmpc}.
    \end{proof}
    
    \PROB{2} can be viewed as a problem that computes the reference sequence $\Theta$ that reproduces the optimal $U$ of \PROB{3} as an unconstrained solution for the lower-level MPC. 
    Indeed, it is also possible to first solve \PROB{3}, then solve \eqref{eq:grad} to obtain $\Theta$, and feed that into \PROB{L}.
    Therefore, if there is no need for a hierarchical design, directly using \PROB{3} is the best option.
    Crucially, even when a hierarchical implementation is desired (for purposes such as reusing existing controllers and dividing development tasks), one can achieve the same performance as \PROB{1} by using \PROB{2} and provide the stability guarantee for the bilevel implementation \PROB{1} or \PROB{2} with \PROB{L} by checking \PROB{3}.

    \begin{remark}[Implicit assumption]
        The above results assume that the underlying optimization problems are solved to the global optima.
        In \PROB{2} and \PROB{3}, this can be achieved by off--the--shelf convex solvers, under Assumptions~\ref{ass:plant}, \ref{ass:nonsing}, and \ref{ass:upper}. 
        In contrast, the original bilevel formulation \PROB{1} yields a nonconvex problem that is sensitive to initialization and prone to spurious local minima. 
        In practice, the returned solution may deviate from the behavior anticipated by Theorems~\ref{thm:subset} and \ref{thm:equiv} and Corollary~\ref{cor:stability}. 
        Section~V demonstrates this empirically.
        \rev{The same issue arises if the numerical solutions of \PROB{2} and \PROB{3} are suboptimal, which remains a topic for future research.}
    \end{remark}

\section{Hierarchical--to--Bilevel Connection}\label{sec:bridge}

    This section connects bilevel MPC and standard HMPC through move-blocking, which reduces the dimension of the decision variable in MPC. 
    This demonstrates that bilevel MPC is a natural extension of HMPC and enables the design of trade-offs between computational efficiency and performance with certified bounds.
        
\subsection{HMPC as a One-Block Reference Policy}

    We first recall a standard HMPC baseline that solves a steady-state optimization:
    \begin{align}\label{eq:P0}
        \!\!\mathbf{P_0}\!: \ &
        \theta^* \!=\! \arg\min_{\theta} \left\{ \ell(x,u)\, \middle|\, \mat{ g_l(x,u)\\ g_u(x,u) } \!\le\! 0 \right\} \!, \mat{x \\ u} \!=\! Z \theta \!\!
    \end{align}
    where $\ell$ is the cost function and $g_l,\, g_u$ are the constraints for the lower and upper levels, respectively.
    Its solution is given to the lower-level MPC \PROB{L} by stacking it as
    \begin{equation}\label{eq:M1}
        \Theta_0 := M_1 \theta^\ast, \
        M_1 := \mathbf{1}_N \otimes I_m \in \R^{Nm\times m}.
    \end{equation}
    The matrix $M_1$ represents the one-block (constant-in-horizon) blocking, so any $\theta \in \R^m$ generates a stacked reference $\Theta = M_1 \theta$. 
    Therefore, \PROB{0} can be seen as a problem of selecting $\Theta \in \im M_1$; i.e., HMPC is a one-block reference policy.
    Here, we consider measuring the \emph{transient} performance of \PROB{0} with the cost function $F_u$ for \PROB{1}--\PROB{3}:
    \begin{align*}
        V_0(x_0) &:= F_u(U_0(x_0);x_0), \\
        U_0(x_0) &:= \arg\min_{U}\{F_l(U;\Theta_0,x_0)\mid G_l(U;x_0)\le 0\}.
    \end{align*}
    Although \PROB{0} has a different objective function $\ell$ than $F_u$, this can be justified if 
    we define $F_u(U;x_0) = \sum_{k=0}^{N-1}\ell(x_k,u_k)$, which is the natural extension of $\ell$ to measure the transient performance.
    We can also add the terminal cost to it to ensure convergence of \PROB{1}--\PROB{3} to the optimal steady-state of \PROB{0} formally (see Appendix~III~\cite{Moriyasu2026bmpc}).
        

\subsection{Move-Blocked Bilevel Problems}

    \begin{assumption}[Blocking matrix]\label{ass:blocking}
        The blocking matrix $M\in\R^{Nm\times p}$ has full column rank, i.e., $\rank(M)=p$.
    \end{assumption}\noindent
    For the above $M$, restricting $\Theta$ to the subspace $\im M$ via
    \[
    \Theta = M\Phi,\  \Phi\in\R^{p},
    \]
    leads the move-blocked version of \PROB{1} and \PROB{2}:
    \begin{align}
        \mathbf{P_1}(M): \ & \min_{\Phi,\,U} \ F_u(U;x_0) \ \text{s.t.} \ G_u(U;x_0)\le 0,\label{eq:P1M}\\ 
        & U = \arg\min_{V}\{\, F_l(V;M\Phi,x_0)\mid G_l(V;x_0)\le 0\,\}, \nonumber 
        \end{align}
        \begin{align}
        \mathbf{P_2}(M): \ & \min_{\Phi,\,U} \ F_u(U;x_0) \ \text{s.t.} \ G_u(U;x_0)\le 0,\label{eq:P2M}\\ & G_l(U;x_0)\le 0,\ \ \nabla_U F_l(U;M\Phi,x_0)=0.\nonumber
    \end{align}
    For a fixed $x_0$, let $V_i(M)$ denote the optimal value of $\mathbf{P}_i(M)$, respectively.

    HMPC \PROB{0} computes a fixed reference and stacks it as $\Theta_0 \in \im M_1$.
    In contrast, \PROB{1}$(M)$ and \PROB{2}($M$) optimize transient behavior through $F_u(U;x_0)$ while still restricting reference to a structured subspace $\im M$.
    By enlarging $\im M$ (relaxing blocking), we obtain a continuous interpolation from HMPC-like one-block references to full bilevel MPC.

\subsection{Optimal-Value Ordering}

    The order of the optimal values for \PROB{0}--\PROB{3} and move-blocked variants \PROB{1}($M$), \PROB{2}($M$) is summarized as:
    \begin{theorem}[Optimal-value ordering]\label{thm:ordering_MB}
    Suppose Assumption~\ref{ass:blocking} holds for any blocking matrices below. 
    Under Assumptions~\ref{ass:plant},~\ref{ass:nonsing} and a fixed $x_0$, the following hold. 
    \begin{enumerate}[\labelindent=0pt]
        \item $V_1 = V_2 \le V_1(M) \le V_2(M)$ for any $M$.
        \item $V_i(\hat M) \le V_i(M) \ (i=1,2)$ for any $M,\hat M$ satisfying $\im(M)\subseteq \im(\hat M)$. 
        \item $V_1(M_1) \le V_0$ if $G_u(U_0;x_0) \le 0$.
        \item $V_2(M_1) \le V_0$ if $G_u(U_0;x_0) \le 0, \nabla_U F_l(U_0;\Theta_0,x_0)\!=\!0$.
    \end{enumerate}
    \end{theorem}    
    \begin{proof} \ \\
    1) The equality $V_1=V_2$ follows from Theorem~\ref{thm:subset}. 
    The move-blocking restricts reference candidates, hence $V_1\le V_1(M),\, V_2\le V_2(M)$.
    If $(\Phi,U)$ is feasible for \PROB{2}($M$), then $G_l(U;x_0)\le 0$ and $\nabla_U F_l(U;M\Phi,x_0)=0$ hold,
    which implies that $U$ is the (unique) lower-level optimizer at $(M\Phi,x_0)$; thus $(\Phi,U)$ is feasible for \PROB{1}($M$) and $V_1(M)\le V_2(M)$ holds.\vspace{3pt}\\
    2) $\im(M)\subseteq \im(\hat M)$ implies the existence of $T$ such that $M=\hat M T$,\footnote{Inclusion $\im(M)\subseteq \im(\hat M)$ can be confirmed by the equivalent condition: $\rank \hat M = \rank [M, \hat M]$. In this case, one matrix $T$ s.t. $M = \hat M T$ can be obtained by $T = \hat M^\dagger M$. } hence the feasible set under $M$ is contained in that under $\hat M$. \vspace{3pt}\\
    3) By construction, we have $\Theta_0 = M_1\theta^\ast$ and $U_0 = \arg\min_{U}\{\, F_l(U;\Theta_0,x_0)\mid G_l(U;x_0)\le 0\,\}$, hence $(\theta^\ast,U_0)$ satisfies the bilevel constraint in \PROB{1}($M_1$).
    Together with $G_u(U_0;x_0)\le 0$, it is feasible and yields $V_1(M_1) \le V_0$.\vspace{3pt}\\
    4) In the proof of (3), if $\nabla_U F_l(U_0;\Theta_0,x_0)=0$ also holds, $(\theta^\ast,U_0)$ satisfies the constraint in \PROB{2}($M_1$). Similarly, $V_2(M_1) \le V_0$ holds.
    \end{proof}\noindent
    Note that if Assumption~\ref{ass:upper} is additionally satisfied, $V_1 = V_2 = V_3 \le V_1(M) \le V_2(M)$ holds from Theorem~\ref{thm:equiv}.
   
    The control sequence $U_0$ from \PROB{0} does not guarantee the upper-level constraint satisfaction $G_u(U_0;x_0) \le 0$.
    This fact motivates the extension from \PROB{0} to \PROB{1}($M$) or \PROB{2}($M$).
    Even if \PROB{0} satisfies $G_u(U_0;x_0) \le 0$, we have $V_1(M_1) \le V_0$, showing the upper-dominance of \PROB{1}($M_1$).
    Furthermore, $\nabla_U F_l(U_0;\Theta_0,x_0)=0$ is required to have $V_2(M_1) \le V_0$. 
    However, we can show the existence of a low-rank $M$ such that $V_2(M) \le V_0$ holds in a certain operational range of $x_0$.

    \begin{proposition}[Low-rank $M$ satisfying $V_2(M)\le V_0$ on region(s)]
    \label{prop:exist_lowrankM_design}
    Under Assumptions~\ref{ass:plant} and \ref{ass:nonsing}, assume that $\{\mathcal R_j\}_{j=1}^J$ satisfy $G_u(U_0(x_0);x_0)\le 0$ for all $x_0\in\cup_{j=1}^J\mathcal R_j$ and define $\mathcal S_{\mathcal R_j}:=\{\Theta^\star(U_0(x_0);x_0)\mid x_0\in\mathcal R_j\}$.
    \begin{enumerate}[\labelindent=0pt]
        \item There exists $M$ such that $\cup_{j=1}^J\mathcal S_{\mathcal R_j}\subseteq\mathrm{im}(M)$, and any such $M$ guarantees \vspace{-6pt}
        \begin{align*}
            V_2(M;x_0) \le V_0(x_0) \  \forall x_0\in\bigcup_{j=1}^J\mathcal R_j.
        \end{align*}\vspace{-5pt}
        \item In case $\{\mathcal R_j\}_{j=1}^J$ is chosen so that the realized control law of \PROB{0} is affine for each region: $U_0(x_0)=A_j x_0+b_j\ \forall x_0\in\mathcal R_j$, one may choose $M$ with \vspace{-2pt}
        \begin{align*}
            p = \mathrm{rank}(M) \le J(n+1).
        \end{align*}
    \end{enumerate}
    \end{proposition}
    \begin{proof} The proof follows from the piecewise affine nature of the control law of \PROB{0}. The full details are provided in Appendix~IV~\cite{Moriyasu2026bmpc}.
    \end{proof}
    In practice, since the spans usually overlap across regions, the required $p$ can be much smaller than $J(n+1)$.
    The actual $M$ with the minimal rank can be obtained by Algorithm~\ref{alg:construct_M_from_P0}, which ensures $V_2(M)\le V_0$ on the sampled \emph{points}.
    To ensure it in the \emph{regions} that include sampled points, one needs $n+1$ affine independent samples from each region.
    
    \begin{algorithm}[t]
    \caption{Construction of $M$ from \PROB{0} data}
    \label{alg:construct_M_from_P0}
    \begin{algorithmic}[1]
        \REQUIRE Samples $\{x_0^{(i)}\}_{i=1}^{N_s}\subset\mathcal \cup_{j=1}^J\mathcal R_j $, tolerance $\texttt{tol}>0$.
        \ENSURE A matrix $M$ satisfying $V_2(M)\le V_0$ on the sampled points.
        
        \STATE For each sample $x_0^{(i)}$:
        \STATE \hspace{0.8em} Run \PROB{0} to obtain $U_0^{(i)}:=U_0(x_0^{(i)})$.
        \STATE \hspace{0.8em} Compute $\Theta^{(i)}:=\Theta^\star(U_0^{(i)};x_0^{(i)})$.
        \STATE Form $\mathbf\Theta := [\,\Theta^{(1)}\ \Theta^{(2)}\ \cdots\, \Theta^{(N_s)}\,]$.
        \STATE Get $M$ by \texttt{M = orth($\mathbf\Theta$,tol)} (MATLAB).
    \end{algorithmic}
    \end{algorithm}
    
\subsection{Degradation Bound}

    We provide inequalities that bound the performance degradation due to the simplification of the problem by considering a generic constrained optimization problem:
    \begin{proposition}[Optimal-value gap for generic optimization problem]\label{thm:cert_nested_set}
        Consider an optimization problem:
        \begin{align} \label{eq:general_opt}
            \min_{w\in \R^q}\ f(w) \   \text{s.t.}\   g(w)\le 0. 
        \end{align}
        Suppose that $f$ is $C^1$ and $H$-strongly convex\footnote{Satisfies $f(v) \ge f(w)+\nabla f(w)^\top (v-w)+\frac12\|v-w\|_{H}^2 \ \forall w,v \in \R^q$.} for some $H\in\PD^q$; and $g$ is $C^1$ and convex.
        Let $V$ be the optimal value and $\mathcal{I}:=\{i\mid g_i(w)=0\}$ be the active set, then
        \begin{align}\label{eq:gap_general}
         0 & \le f(w) - V \\
           & \le \tfrac12 \min_{\mu\ge 0} \norm{ \nabla f(w) + \nabla g_{\mathcal{I}}(w) \mu }^2_{H^{-1}} =: \Delta(w),\nonumber
        \end{align}
        holds for any $w\in C:= \Set{ w | g(w)\leq 0 }$ satisfying Mangasarian--Fromovitz constraint qualification (MFCQ)\footnote{There exists $d$ such that $\nabla g_\mathcal{I}(w)^\top d < 0$. Linear independent constraint qualification (LICQ) or Slater condition is sufficient for MFCQ.}.
    \end{proposition}    
    \begin{proof}
 Trivially, we have $f(w) - V \geq 0$.
        Let $v = w^* = \arg\min_{w\in C} f(w)$, then the condition $f(v) \ge f(w)+\nabla f(w)^\top (v-w)+\frac12 \norm{v-w}_{H}^2 \ \forall w,v \in \R^q$ yields $f(w) - V \leq - (\frac{1}{2} \norm{d}_H^2 + \nabla f(w)^\top d) \ \forall w \in C$, where $d:= w^* - w$.
        Introducing the tangent cone $T_C(w) \ (\supseteq C - \{w\} \ni d)$ yields
        \begin{align}
            0 & \geq V- f(w)
                \geq \inf_{d \in C-\{w\}} \left\{\tfrac12 \norm{d}_H^2 + \nabla f(w)^\top d \right\}\nonumber\\ 
              & \geq \inf_{d \in T_C(w)} \left\{\tfrac12 \norm{d}_H^2 + \nabla f(w)^\top d\right\}. \label{eq:etaineq}
        \end{align}
        Under MFCQ at $w$, $T_C(w) = \Set{ d \in \R^q | \nabla g_{\mathcal{I}}(w)^\top d \le 0 }$.
        Therefore, the right-hand side reduces to the convex QP
        \begin{equation*}
            \min_{d}\ \tfrac12 \norm{d}^2_H + \nabla f(w)^\top d \ \text{s.t.} \ \nabla g_{\mathcal{I}}(w)^\top \,d \le 0 ,
        \end{equation*}
        and it is lower-bounded by its Lagrange dual:
        \[
         \max_{\mu\ge 0} -\tfrac12 \norm{ \nabla f(w) + \nabla g_\mathcal{I}(w) \mu }^2_{H^{-1}} = -\Delta(w),
        \]
        which proves the claim, combining it with (\ref{eq:etaineq}).
    \end{proof}
    \noindent
    The computation for evaluating $\Delta(w)$ is inexpensive because the problem reduces to a nonnegative least-squares problem with dimension $|\mathcal{I}|$, which is usually much smaller than $q$.
    
    \begin{remark}[Tighter bound]\label{rem:tighten_near_active}
        The bound $\Delta(w)$ can be conservative when only a few constraints are active (e.g., under coarse blocking such as $M_1$).
        To evaluate a tighter bound, one can use
        \begin{align*}
            \Delta_\ep(w)
            & :=\min_{\mu\ge 0}
            \PAR{ \tfrac12\norm{ \nabla f(w)\!+\!\nabla g_{\mathcal J}(w)\mu }^2_{H^{-1}} \!-\! g_{\mathcal J}(w)^\top \mu},\\
            {\mathcal J}(w;\ep) &:= \Set{ i | g_i(w) \ge - \ep } \ (\ep>0),
        \end{align*}
        instead of $\Delta(w)$.
        Its detail is provided in Appendix~V~\cite{Moriyasu2026bmpc}.
    \end{remark}

    Utilizing this, we provide the upper-bounds of performance degradation.
    
    \begin{theorem}[\emph{A posteriori} bounds for degradation]\label{thm:cert_nested_U}
        For a fixed $x_0$, suppose that Assumptions~\ref{ass:plant}, \ref{ass:nonsing}, and \ref{ass:upper} hold; $F_u(\cdot;x_0),\, \allowbreak G_u(\cdot;x_0)$ are $C^1$.
        Let $M, \hat M$ satisfy Assumption~\ref{ass:blocking} and $\im M \subseteq \im \hat M$; and $T$ satisfies $M = \hat M T$.
        Set (\ref{eq:general_opt}) as: 
        \begin{align}\label{eq:set_for_P2gap}
            f(w) &= F_u(U^\star( \hat M w ;x_0);x_0), \\
            g(w) &= G(U^\star( \hat M w;x_0);x_0), \ w = \hat \Phi, \nonumber 
        \end{align}
        where $U^\star(\Theta;x_0) := \bigl( \bar B^\top \bar Q\bar B + \bar R \bigr)^{-1} ( \bar \Gamma \Theta -  \bar B^\top \bar Q \bar A x_0),\,\allowbreak G(U;x_0):=[G_l(U;x_0); G_u(U;x_0)]$.
        \begin{enumerate}[\labelindent=0pt]
            \item If \PROB{2}($M$) is feasible and satisfies MFCQ at its optimal solution $\Phi^\star$, then the following holds.
                \[
                    0 \le V_2 (M) - V_2 (\hat M) \le \Delta (T \Phi^\star).
                \]
            \item Consider $\hat M=I_{Nm}$ (therefore $T=M$). 
            If \PROB{0} is feasible and its solution $U_0(x_0)$ satisfies $G_u(U_0;x_0)\leq 0$ and MFCQ for \eqref{eq:set_for_P2gap} at $w = \Theta^\star(U_0;x_0)$ (see \eqref{eq:Thetastar}), then the following holds.
                \[
                    0 \le V_0 - V_2 \le \Delta \PAR{\Theta^\star(U_0(x_0);x_0)}.
                \]
        \end{enumerate}
    \end{theorem}
    \begin{proof} The proof follows from applying Proposition~\ref{thm:cert_nested_set}. The details are provided in Appendix~VI~\cite{Moriyasu2026bmpc}.
    \end{proof}
    
    The bounds in Theorem~\ref{thm:cert_nested_U} are designed for one-sided evaluation: it suffices to solve only the ``efficient'' tractable problem (\PROB{2}($M$) and \PROB{0} for cases (1) and (2) of Theorem~\ref{thm:cert_nested_U}, respectively) and then evaluate $\Delta(\cdot)$ at the resulting candidate point. 
    Importantly, the ``expensive'' counterpart (\PROB{2}($\hat M$) and \PROB{2} for cases (1) and (2) of Theorem~\ref{thm:cert_nested_U}, respectively) does \emph{not} need to be solved to certify a performance gap.
    
    After computing the ``efficient'' problem, the inequality provides a quantitative limit on the maximum performance improvement achievable by considering the transient behavior (case (2)) or enlarging the blocking subspace $\im(M)\subseteq\im(\hat M)$ (case (1)). 
    This enables a principled ``stop-or-refine'' decision: if the bound is already small, increasing degrees of freedom or switching to the heavier formulation is guaranteed to yield only marginal improvement.
    
    Finally, computing $\Delta(\cdot)$ is inexpensive. 
    It amounts to solving a nonnegative least-squares problem in the multipliers associated with the active inequalities at the evaluation point, which is typically much smaller than the original problem. 
    Hence, the proposed bounds act as lightweight, practically computable certificates with negligible overhead.
    \rev{Note that the bounds depend on the initial state $x_0$, and that move-blocking generally does not preserve stability. These are important challenges for future research.}

\section{Numerical Evaluation}\label{sec:num}

Here, we present a toy example to verify the validity of the theory with minimal complexity. An example demonstrating the applicability to realistic problems, using a quadrotor as the target, is provided in Appendix~VII~\cite{Moriyasu2026bmpc}.

\subsection{Settings}
\subsubsection{Target system}
    We compare the solutions of problems \PROB{0}--\PROB{3} and their variants for a simple linear system:
    \begin{align*}
        x_{k+1} &= A x_k + B u_k, \ x = \mat{p;v}\in \R^2, \ u \in \R, \\
        A &= \mat{1 & T_s\\ 0 & 1}, \ B = \mat{ \frac{1}{2} T_s^2 \\ T_s }, \ x_0= \mat{-1\\0},
    \end{align*}
    where $T_s=0.3~$s is the sampling period, $p$ is position, and $v$ is velocity. 
    This satisfies Assumption~\ref{ass:plant} ($\rank (S)=2$).
    
\subsubsection{Lower-level MPC}
    The lower-level tracking MPC \PROB{L} (\ref{eq:lower}) is constructed according to Example~\ref{ex:example1} with $(Q,P,R)=(I,I,0.1\,I),\, N = 10$, which makes Assumption~\ref{ass:nonsing} hold.
    Although $I-A$ is not invertible for this system, its steady-state can be parameterized by $[\bar x; \bar u] = [Z_x; Z_u] \theta, \, Z_x = [1;0], \, Z_u = 0$, implying $\theta$ means the steady position.
    We set the input box constraint $-{\bf 1}_N\le U\le{\bf 1}_N$ encoded via $G_l(U;x_0)\le 0$, which is the stack of stage-wise constraints $g_l(x_k,u_k)=[u_k-1;-u_k-1]\leq 0$.
    Due to strong convexity, this problem has a unique solution for any $x_0$.

\subsubsection{Upper-level concept}
    As an upper-level goal, we assume a fixed target state $\check x = [1;0]$ and input $\check u = 0$, and the state (position) bound $p \leq 0$ encoded as $g_u(x,u) \leq 0$.
    Therefore, the target state is outside the feasible set.
    Under this position constraint and the lower-level input constraint, we aim to minimize the upper-level cost: 
    \[
        \ell(x,u) := \tfrac12 \norm{x - \check x}_Q^2 + \tfrac12 \norm{u - \check u}_R^2.
    \]
    Here, we used the same $Q$ and $R$ as the lower-level, for simplicity.

\subsubsection{Hierarchical MPC}
    Before considering a bilevel MPC, we construct HMPC \PROB{0} \eqref{eq:P0} as a baseline, using the above $\ell, \, g_l$ and $g_u$.
    For this problem setup, the solution of the steady-state optimization is $\theta^*=0$, meaning $[\bar x; \bar u] = [0;0]$.

\subsubsection{Bilevel and centralized MPCs}
    We define a bilevel MPC \PROB{1} \eqref{eq:P1}, its reduction \PROB{2} \eqref{eq:P2}, their move-blocked variants \PROB{1}($M$) \eqref{eq:P1M} and \PROB{2}($M$) \eqref{eq:P2M}, and the corresponding centralized MPC \PROB{3} \eqref{eq:P3}, using:
    \begin{align*}
        F_u(U;x_0) &= V_f(x_N) + \sum_{k=0}^{N-1} \ell(x_k,u_k), \\
        G_u(U;x_0) &= \mat{ g_u(x_0,u_0) \\ \vdots \\ g_u(x_{N-1},u_{N-1}) }, \
        V_f(x) := \tfrac12 \norm{ x - \check x }^2_{P},
    \end{align*}
    which satisfies Assumption~\ref{ass:upper} for the initial state $x_0$ shown above.
    To examine the effect of move-blocking, we use
    \begin{align}\label{eq:Mi}
        M_i := \mat{I_i\\ \mat{O ,\mathbf{1}_{N-i}}} \otimes I_m \in \R^{Nm\times im} \ (i\in[1,N]),
    \end{align}
    that satisfies Assumption~\ref{ass:blocking} and makes the first $i$ steps in the horizon free and the last $(N\!-\!i)$ steps blocked.
    This also satisfies $\im M_i \subseteq \im M_{i+1}$ for any $i\in[1,N\!-\!1]$ and corresponds to $M_1$ in \eqref{eq:M1} in case $i=1$.
    In addition, $i=N$ (in this case, $i=10$) means that the problem has full degrees of freedom, i.e., \PROB{1}($M_N$)/\PROB{2}($M_N$) are equivalent to \PROB{1}/\PROB{2}, respectively.
    A low-rank move-blocking matrix $M^\star$ is also prepared by generating 50 initial states around $x_0$ and giving it to Algorithm~\ref{alg:construct_M_from_P0}.
    The obtained $M^\star$ had $\rank (M^\star)=2$.

\subsubsection{Numerical methods}
    \PROB{1} is converted into an equivalent single-level problem:
    \begin{align*}
    &\min_{\Theta}\ F_u(U^*(\Theta;x_0);x_0) \ \text{s.t.} \ G_u(U^*(\Theta;x_0);x_0)\leq 0,\text{where} \\
    & \ U^*(\Theta;x_0) := \arg\min_{V} \Set{ F_l(V;\Theta,x_0) | G_l(V;x_0)\le 0 },
    \end{align*}
    and is solved by SQP solver (MATLAB \texttt{fmincon}) using implicit differentiation to obtain $\nabla_{\Theta} U^*$.
    \PROB{1}($M$) is also converted in the same manner.
    \PROB{0}, \PROB{2}, \PROB{2}($M$), and \PROB{3} are solved by QP solver (MATLAB \texttt{quadprog}).
    
    \begin{figure}[t]
        \centering
        {
        \includegraphics[trim={40 0 50 0}, clip, scale=0.5]{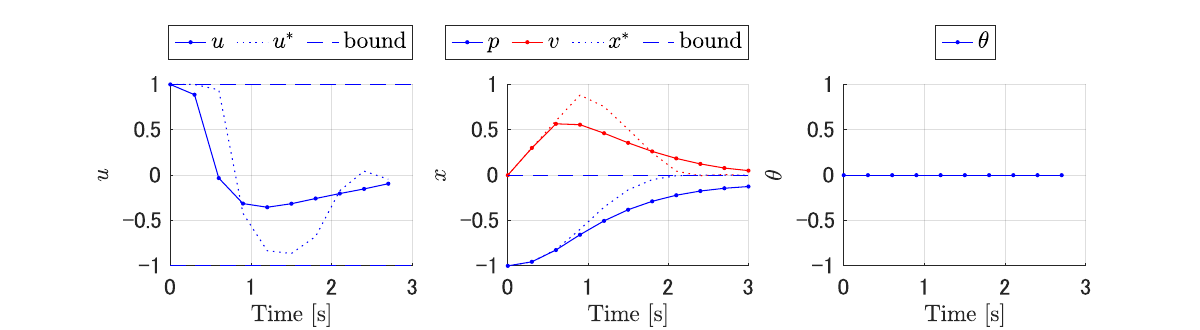}
        \subcaption{\PROB{0} \label{fig:P0}}
        \includegraphics[trim={40 0 50 0}, clip, scale=0.5]{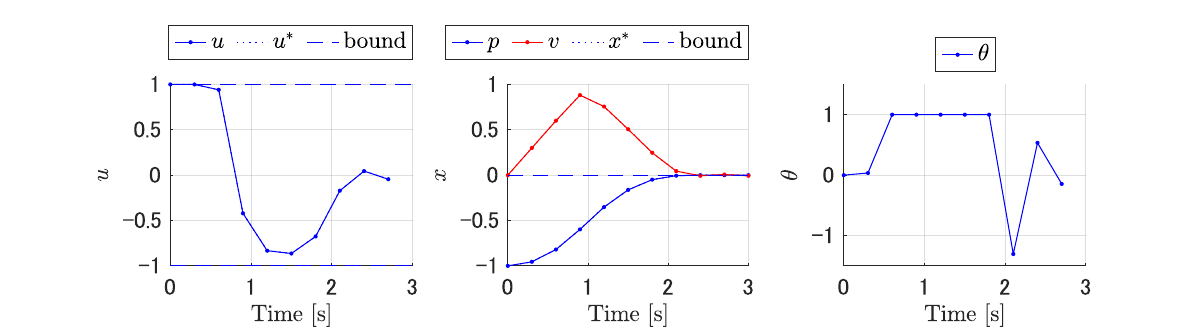}
        \subcaption{\PROB{1} with initial solution $\Theta = 0$.\label{fig:P1-0}}
        \includegraphics[trim={40 0 50 0}, clip, scale=0.5]{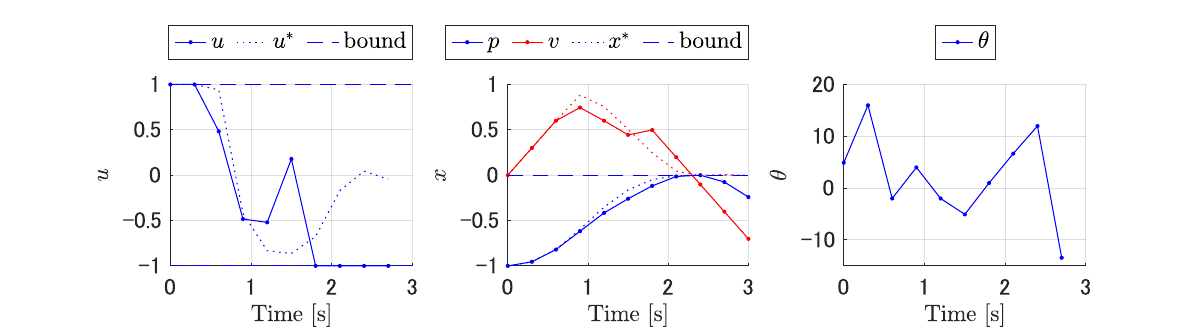}
        \subcaption{\PROB{1} with random initial solution.\label{fig:P1-rand}}
        \includegraphics[trim={40 0 50 0}, clip, scale=0.5]{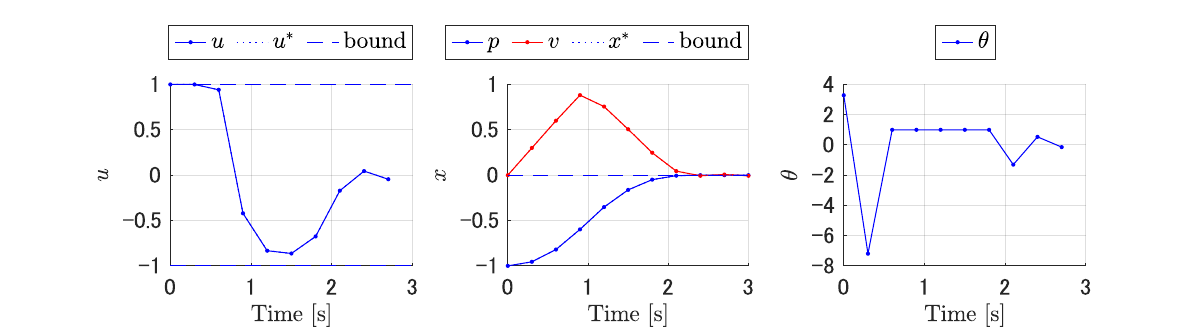}
        \subcaption{\PROB{2} \label{fig:P2}}
        }
        \caption{Numerical comparison of \PROB{0}, \PROB{1}, and \PROB{2}. \PROB{2} reliably reproduces the centralized solution \PROB{3} (dotted lines), whereas \PROB{1} does so only under favorable initialization, highlighting the practical robustness of the proposed reduction.\label{fig:toysol}}
    \end{figure}
    
\subsection{Results}
    
    Fig.~\ref{fig:toysol} shows the numerical solutions for \PROB{0}, \PROB{1}, and \PROB{2}.
    Fig.~\ref{fig:P0} and \ref{fig:P2} are for \PROB{0} and \PROB{2}, respectively, and Fig.~\ref{fig:P1-0} and \ref{fig:P1-rand} are for \PROB{1} obtained with (b) $\Theta=0$ initialization and (c) random $\Theta$ initialization.
    In each figure, colors distinguish the elements of the state (blue: $p$, red: $v$); the results for each problem are shown in solid lines; the optimal sequence of \PROB{3} (denoted as $u^*,\, x^*$) is shown in dotted lines as a baseline; the upper and lower bounds are shown in dashed lines.
    From these results, we can confirm:
    \begin{itemize}[\labelindent=0pt]
        \item \PROB{0} (Fig.~\ref{fig:P0}) provided the sequence obtained from \PROB{L} with $\Theta=0$, which was different from that of \PROB{3}, showing that the result was not optimal in the measure of $F_u$.
        \item \PROB{1} with $\Theta=0$ initialization (Fig.~\ref{fig:P1-0}) and \PROB{2} (Fig.~\ref{fig:P2}) provided the optimal sequence that coincided with the result of \PROB{3}. 
        Both had different sequences in $\Theta$, implying there were multiple (actually, infinite) solutions for \PROB{1}.
        \PROB{2} successfully provided one of them as a unique solution.
        \item \PROB{1} with random initialization (Fig.~\ref{fig:P1-rand}) failed to obtain the optimal sequence.
        This showed sensitivity to the choice of the initial solution due to the non-convexity of \PROB{1}, motivating the use of the proposed problem \PROB{2}. 
    \end{itemize}
  
    Fig.~\ref{fig:toyMB} summarizes the performance degradations and their estimated bounds of \PROB{0}, \PROB{1}($M_i$), and \PROB{2}($M_i$) ($i=1,\ldots,10$), compared to \PROB{2}\,($=$\,\PROB{1},\,\PROB{3}).
    In this figure, blue, green, and red circles show the performance degradations $V_2(M_i)-V_2,\, V_1(M_i)-V_2 \ (i=1,\ldots,10)$ and $V_0-V_2$, respectively; stems represent the upper bounds, i.e., blue and red ones show $\Delta(T\Phi^\star),\, \Delta(\Theta^\star)$, respectively; cross markers represent the tighter bounds, i.e., the blue and red ones show $\Delta_\ep (T\Phi^\star)$ and $\Delta_\ep (\Theta^\star)$ with $\ep = 0.5$; and the blue dashed line represents $V_2(M^\star) -V_2$. 
    Note that \PROB{1}($M$) is solved with $\Theta=0$ initialization. 
    From the figure, we can see:

    \begin{itemize}[\labelindent=0pt]
        \item The results are consistent with the orders shown in Theorem~\ref{thm:ordering_MB}, i.e., $V_0 \geq V_1(M_1), \, V_j(M_i) \geq V_j(M_{i+1}) \ (i=1,\ldots,9, \, j=1,2), \, V_2(M_i) \geq V_1(M_i) \ (i=1,\ldots,10)$.
        \item Although $V_0 \geq V_2(M_1)$ did not hold because $\nabla F_l\neq 0$ at the solution of \PROB{0}, $V_0 \geq V_2(M^\star)$ is satisfied as intended.
        \item $\Delta$ was confirmed to work as an upper-bound of degradation for \PROB{2}($M$) and \PROB{0}. It looks looser in small $i$, but $\Delta_\ep$ succeeded in tightening the bound.
    \end{itemize}

    \begin{figure}[t]
        \centering
        \includegraphics[scale=0.75,trim=5 0 0 0]{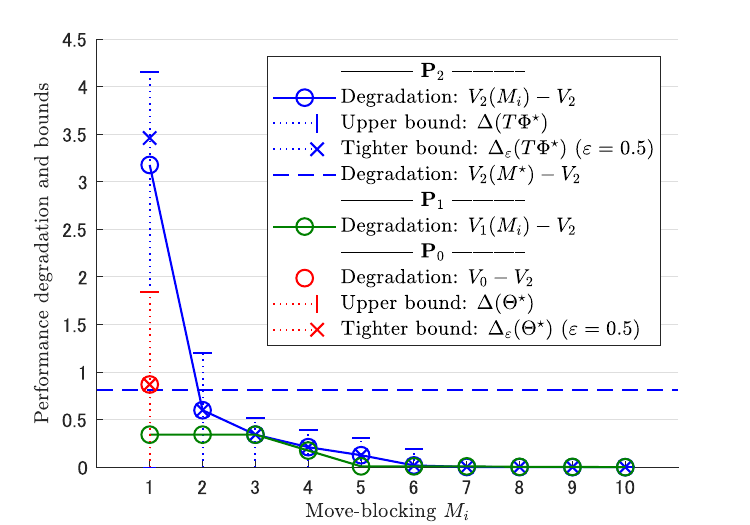}
        \caption{Performance degradation relative to \PROB{2} for move-blocked problems (blue circle: \PROB{2}($M$), green circle: \PROB{1}($M$)) and HMPC \PROB{0} (red circle), together with the upper bounds. The plot shows monotonic performance improvement as the blocking is relaxed; the upper bound (bar) captures the degradation trend well, while the tighter bound (cross) is more effective under coarse blocking; and \PROB{2}($M$) with the low-rank matrix $M^\star$ (broken line) outperforms \PROB{0}.\label{fig:toyMB}}
    \end{figure}   

\section{Conclusions}
We studied a bilevel MPC architecture in which an upper level selects steady-state/reference parameters, while a lower level implements a linear tracking MPC formulated as a strongly convex QP, assuming that the upper objective and constraints depend on the lower-level outcome only through the realized input trajectory.

We derived a smooth single-level reduction that avoids complementarity by lifting lower-level inequalities to the upper level and enforcing only stationarity. Under a verifiable nonsingularity condition, the reformulation is exact and recovers the same optimal value and realized input trajectory as the original bilevel MPC. Under standard convexity assumptions, the realized input optimizer is unique and coincides with that of the associated centralized MPC, enabling certification of closed-loop properties.

Finally, we connected HMPC and bilevel MPC via move-blocking, establishing optimal-value ordering among hierarchical, bilevel, and centralized formulations and their blocked variants. This connection provides computable certificates for blocking-induced degradation and enables a principled trade-off between performance and computational complexity. Numerical results on an integrator-containing toy example illustrate the recovery of the centralized solution and the effectiveness of certificate-guided move-blocking.

These results make bilevel control architectures, often avoided in practice due to their complexity, more practical to deploy. They also reduce reliance on ad hoc heuristics and support a more systematic design process across a broad range of applications.

\printbibliography

\appendices

\rev{
\section{Nonsingular Design Possibility}\label{app:nonsing_design}

    \begin{proposition}[Generality of nonsingularity]\label{prop:nonsing_design} \ \\
    Under Assumption~\ref{ass:plant},
    \begin{enumerate}[\labelindent=0pt]
        \item there exist $(\bar Q, \bar R) \in \PSD^{Nn} \times \PD^{Nm}$ s.t. Assumption~\ref{ass:nonsing} holds;
        \item the random choice of $(\bar Q,\bar R) \in \PSD^{Nn} \times \PD^{Nm}$ satisfies Assumption~\ref{ass:nonsing} with probability $1$.
    \end{enumerate}
    \end{proposition}
    \begin{proof}
        1) Define the stage-wise map $\Gamma (Q,R) := B^\top Q Z_x + R Z_u \in\R^{m\times m}$ for $Q \in \PSD^n$ and $R \in \PD^m$.
        From Lemma~\ref{lem:stage_nonsing} shown below, there exist $Q \in \PSD^n$ and $R \in \PD^m$ s.t. $\Gamma (Q,R)$ is nonsingular.
        Using the block-diagonal setting in Example~\ref{ex:example1}, i.e., $\bar Q = \mathrm{blkdiag}(Q,\dots,Q,P)$ and $ \bar R = I_N\otimes R$, the matrix $\bar \Gamma (\bar Q,\bar R) = \bar B^\top \bar Q {\bar Z}_x + \bar R {\bar Z}_u$ becomes block upper-triangular with diagonal blocks $\Gamma (Q,R)$ (for stages $1,\dots,N-1$) and $\Gamma (P,R)$ (at the terminal stage).
        Choosing $Q,P,R$ s.t. $\Gamma (Q,R)$ and $\Gamma (P,R)$ are nonsingular, $\bar \Gamma (\bar Q,\bar R)$ becomes nonsingular.

        2) Consider expanding $(\bar Q,\bar R)$ to $\bbS^{Nn} \times \bbS^{Nm}$. Assumption~\ref{ass:nonsing} is satisfied for any $(\bar Q,\bar R)$ except on the hypersurface defined by $\det \Gamma(\bar Q,\bar R)=0$, which is a polynomial in the entries of $\bar Q$ and $\bar R$.
        Since this polynomial is not identically zero by item~1, its zero set has Lebesgue measure zero \cite{BasuPollackRoy2006}, i.e., the random choice of $(\bar Q,\bar R)\in \bbS^{Nn} \times \bbS^{Nm}$ satisfies Assumption~\ref{ass:nonsing} with probability $1$.
        The same result holds for $(\bar Q,\bar R) \in \PSD^{Nn} \times \PD^{Nm}$, since it has a nonempty interior in $\bbS^{Nn} \times \bbS^{Nm}$.
    \end{proof}

    \begin{lemma}\label{lem:stage_nonsing}
        Under Assumption~\ref{ass:plant}, there exist cost matrices $(Q,R) \in \PSD^n \times \PD^m$ s.t. $\Gamma (Q,R) := B^\top Q Z_x + R Z_u \in\R^{m\times m}$ is nonsingular for every basis matrix $Z=[Z_x;\, Z_u] \in \R^{(n+m)\times m}$ of $\ker S$.
    \end{lemma}
    
    \begin{proof}
    The nonsingularity of $\Gamma(Q,R)=[ B^\top Q,  R]Z$ reduces to the bijectivity of the basis-independent linear map
    \[
        C_{Q,R}: \ker S\ni [x;u] \mapsto [ B^\top Q,  R] [x;u] \in \R^m,
    \]
    since $\im Z = \ker S$ holds for any basis matrix $Z$.
    
    Define ${\mathcal V}\!:=\! \ker (I\!-\!A), \, l\!:=\! \dim {\mathcal V}$, then, for the case with $l\!=\!0 \, (\Leftrightarrow 1 \!\notin\! \eig A)$, $(Q,R)\!=\!(O,I_m)$ makes $C_{Q,R}$ bijective.
    Indeed, if $l\!=\!0$, then $\ker S \!=\! \im [ (I\!-\!A)^{-1}B; I_m]$ holds, and hence, $C_{Q,R} \ker S \!=\! \im [ O, I_m] [ (I\!-\!A)^{-1}B; I_m] = \im I_m = \R^m$.
    In this case, $\R^m$ is fully spanned by the term $R u$.

    The following considers the case with $l>0 \, (\Leftrightarrow 1 \!\in\! \eig A)$, where $\ker S$ contains a $l$-dimensional subspace ${\mathcal V} \times \{0\}$ that is invisible from $Ru$. 
    The main idea of the proof is to make this subspace visible from $B^\top Qx$.
    
    We define linear subspaces and mappings to show the connection among $\R^n, \R^m$ and the input-invisible space $\R^l$.
    Since $\dim \im (I-A) = n - l$ and $\im (I-A)+\im B=\R^n$ hold by Assumption~\ref{ass:plant}, there exists a subspace ${\mathcal W} \!\subseteq\! \im B$ s.t. 
    \[
        \R^n=\im (I-A)\oplus {\mathcal W}, \ l=\dim{\mathcal W} \ (\le \dim \im B \le m).
    \]
    The equal-dimensional subspaces ${\mathcal V}$ and ${\mathcal W}$ admit a common complement ${\mathcal N}$: $\R^n = {\mathcal V} \oplus {\mathcal N} = {\mathcal W} \oplus {\mathcal N}$.
    Let $\Pi:\R^n\to{\mathcal V}$ be the projection onto ${\mathcal V}$ along ${\mathcal N}$, i.e., ${\mathcal V} = \im \Pi, \, {\mathcal N} = \ker \Pi$. Then, its restriction $\Pi|_{\mathcal W}: {\mathcal W} \to {\mathcal V}$ is bijective and ${\mathcal W} = (\Pi|_{\mathcal W})^{-1} {\mathcal V}$ holds.
    From ${\mathcal W} \subseteq \im B$, the restricted Moore-Penrose inverse $B^\dagger|_{\mathcal W}: {\mathcal W} \to {\mathcal Z}:= B^\dagger{\mathcal W}$ is bijective. 
    Define $\Omega: \R^n\to {\mathcal Z}$ by their composition $\Omega := B^\dagger (\Pi|_{\mathcal W})^{-1} \Pi$. 
    In addition, from $l= \dim {\mathcal V}$, there exists a bijective map $L: \R^l \to {\mathcal V}$.
    Fig.~\ref{fig:coordinate-free-construction} summarizes the above.
    Crucially, ${\mathcal V},{\mathcal W},{\mathcal Z}$, and $\R^l$ are all connected via bijective mappings.
    
    The mapping $\Omega$ makes $[\Omega,I_m]$ bijective from $\ker S$ to $\R^m$.
    Indeed, when $[x;u]\in \ker S$ satisfies $[\Omega,I_m][x;u]=0 \ (\Leftrightarrow u = -\Omega x)$, then $(I-A)x= Bu = -B\Omega x$.
    The left-hand side belongs to $\im (I-A)$, whereas the right-hand side is on ${\mathcal W}$.
    Hence, $\R^n=\im (I-A)\oplus {\mathcal W}$ requires both sides to be zero.
    Since $B\Omega = (\Pi|_{\mathcal W})^{-1}\Pi$ and $\ker ((\Pi|_{\mathcal W})^{-1}\Pi) = \ker \Pi = {\mathcal N}$, we have $x \in \ker (I-A) \cap \ker(B\Omega) = {\mathcal V} \cap {\mathcal N}=\{0 \}$ from $\R^n={\mathcal V} \oplus {\mathcal N}$.
    Thus, $x=0, \,u= -\Omega x=0$ holds, showing the bijectivity of $[\Omega,I_m]: \ker S \to \R^m$, as $\dim \ker S= m$.
    
    The following $(Q,R)\in \PSD^n \times \PD^m$ makes $C_{Q,R}$ bijective:
    \[ 
    Q = \Lambda_x^\top \widehat Q \Lambda_x,\quad R = (\Lambda_x B)^\top \widehat Q (\Lambda_x B) + \Lambda_u^\top \widehat R \Lambda_u,
    \]
    where $(\widehat Q,\widehat R) \in \PD^l \times \PD^{m-l}$, and the maps $\Lambda_x: \R^n \rightarrow \R^l,\, \allowbreak\Lambda_u: \R^m \rightarrow \R^{m-l}$ are defined s.t. $\Lambda_x = L^\dagger \Pi, \, \ker \Lambda_u = {\mathcal Z}$.
    Here, $Q\in\PSD^n$ is trivial, and $R\in\PD^m$ follows from $u^\top R u =0 \Leftrightarrow u \in \ker (\Lambda_x B) \cap \ker\Lambda_u = \{0\}$.
    This intersection is $\{ 0 \}$, because $(\Lambda_x B)|_{\mathcal Z} = (L^\dagger \Pi B)|_{\mathcal Z}$ is bijective from ${\mathcal Z}$ to $\R^l$, and when $u\in \ker \Lambda_u = {\mathcal Z}$, then $u\in \ker (\Lambda_x B)$ yields $u=0$.
    Lastly, $\Lambda_xB\Omega=\Lambda_x,\, \Lambda_u\Omega = O$ yield $R\Omega = B^\top \Lambda_x^\top \widehat Q \Lambda_x = B^\top Q,\allowbreak [B^\top Q, R] = R[\Omega, I_m]$, showing the bijectivity of $C_{Q,R}$.
    When $l=m$, then $R$ omits $\Lambda_u^\top \widehat R \Lambda_u$, and $l\le \dim \im B \le m$ gives $\dim \im B = m$ and ${\mathcal Z}=\R^m$, which make the above discussion valid without $\Lambda_u$.
    \end{proof}
}
    \begin{figure}[t]
        \flushright
        \begin{tikzcd}
            \R^n
            \arrow[r, "\Pi", two heads]
            \arrow[d, "\Omega"', two heads]
            \arrow[rd, "B\Omega" sloped, two heads]
            &[0.3cm]
            {\mathcal V} \rlap{$ \ := \ker (I\!-\!A)$}
            \arrow[d, "(\Pi|_{\mathcal W})^{-1}", "\sim"' sloped]
            &
            |[xshift=2.5cm]| \R^l
            \arrow[l, end anchor={[xshift=2cm]}, "L"', "\sim"]
            \\
            \llap{$\R^m \supseteq B^\dagger {\mathcal W} =: \ $} {\mathcal Z}
            &
            {\mathcal W} \rlap{$ \ \subseteq \im B$}
            \arrow[l, "B^\dagger|_{\mathcal W}"', "\sim"]
            &
        \end{tikzcd}
        \caption{Relations among the spaces and mappings.}
        \label{fig:coordinate-free-construction}
    \end{figure}
\section{Proof of Corollary~\ref{cor:stability}}\label{app:equivalence}
    Let $\kappa$ denote the receding-horizon control law induced by $\mathbf{P}_3$ and suppose the closed-loop with \PROB{3}, i.e., $x_{k+1}=Ax_k+B\kappa(x_k)$, has the stated stability property.
    By Theorem~\ref{thm:equiv}, at each time step and for any initial state, the cascade connection of \PROB{1} or \PROB{2} and \PROB{L} produces the same optimizing input sequence in $U$ as \PROB{3}.
    Hence, the implemented first control move is identical to $\kappa(x_k)$.
    Therefore, the closed-loop trajectories coincide, and the claimed stability property is inherited.
    
\section{Convergence Design for MPCs}\label{app:convergence_design}

    \PROB{0} decides the optimal steady state and input pair $(x^*, u^*)$, and the tracking MPC \PROB{L} can ensure convergence to the point by appropriate design.
    The bilevel cascade of \PROB{1} or \PROB{2} and \PROB{L} and its equivalent \PROB{3} can also ensure convergence to the same point under some natural settings.

    Given that the upper-level goal is to minimize $\ell$ under the constraints $g_u\le 0$ and $g_l\le 0$, \PROB{0} is the natural design for the goal when we can ignore the transient behavior.
    On the other hand, when we consider the transient, setting $F_u(U;x_0) = V_f(x_N)+\sum_{k=0}^{N-1}\ell(x_k,u_k)$ and constructing $G_u$ and $G_l$ by stacking $g_u$ and $g_l$ is also natural, and it can be seen as an extension of \PROB{0} to the transient. 
    In this setting, \PROB{3} generally reduces to an economic MPC, and the appropriate terminal cost $V_f$ ensures convergence to $(x^*,u^*)$ that minimizes $\ell(x^*,u^*)$ within the constraints.
    Extensive theories have been studied to ensure the convergence of the economic MPC~\cite{RawlingsMayneDiehl2017}.
    Under Assumptions~\ref{ass:plant}, \ref{ass:nonsing}, and \ref{ass:upper}, the cascade of \PROB{1} or \PROB{2} and \PROB{L} inherits this convergence property from Corollary~\ref{cor:stability}.
    Therefore, all of these controllers can achieve convergence to the optimal steady point of \PROB{0} by appropriate design.
    
    Under such a scenario, using $F_u$ as the measure of the \emph{transient} performance for \PROB{0} is justified.

\section{Proof of Proposition~\ref{prop:exist_lowrankM_design}}\label{app:lowrankM}
    \noindent
    1) $M=I_{Nm}$ trivially satisfies $\cup_{j=1}^J\mathcal S_{\mathcal R_j}\subseteq\mathrm{im}(M) = \R^{Nm}$.
    Under $\cup_{j=1}^J\mathcal S_{\mathcal R_j}\subseteq\mathrm{im}(M)$, for any $x_0\in\cup_{j=1}^J\mathcal R_j$, there exists a solution $\Phi = \Phi^\star \in \R^p$ for $\Theta^\star(U_0(x_0);x_0) = M \Phi$ since $\Theta^\star (U_0(x_0);x_0)\in \im M$.
    This implies $(\Phi^\star,U_0)$ is a feasible solution for \PROB{2}($M$), and therefore, $V_2(M;x_0)\le F_u(U_0;x_0) = V_0(x_0)$ holds for any $x_0\in\cup_{j=1}^J\mathcal R_j$.\vspace{5pt}\\
    2) Since $x_0 \mapsto \Theta^\star(U_0(x_0);x_0)$ is affine in each region, $\dim ( \mathrm{span}(\mathcal S_{\mathcal R_j}) ) \le n+1$ holds for each $j$. 
    This yields $\dim ( \mathrm{span} ( \cup_{j=1}^J \mathcal S_{\mathcal R_j} )) \!\le\! \sum_{j=1}^N \dim (\mathrm{span}(\mathcal S_{\mathcal R_j})) \!\leq\! J(n+1)$.
    Thus, there exists $M$ s.t. $\cup_{j=1}^J\mathcal S_{\mathcal R_j} \subseteq \mathrm{span} ( \cup_{j=1}^J \mathcal S_{\mathcal R_j} ) \subseteq \mathrm{im}(M)$ with $\dim \im M = \rank (M) \leq J(n+1)$.

\section{Tighter Bound in Proposition~\ref{thm:cert_nested_set}}\label{app:tighterbound}

    We can obtain a tighter bound by incorporating \emph{nearly-active} constraints through a slack-preserving linearization. 
    For any $\ep>0$, ${\mathcal J}(w;\ep) := \Set{ i | g_i(w) \ge - \ep }$ satisfies $\mathcal J\supseteq \mathcal I(w)$.
    The convexity of $g$ implies that for any feasible $v\in C$ and $d:=v-w$, $g_{\mathcal J}(v)\ge g_{\mathcal J}(w)+\nabla g_{\mathcal J}(w)^\top d \le 0$, and hence $C-\{w\}\ \subseteq\ \Set{ d\in\R^q | \nabla g_{\mathcal J}(w)^\top d \le -g_{\mathcal J}(w) }$.
    Replacing $T_C(w)$ with this outer approximation yields the tightened bound $0 \le f(w)-V \le \Delta_\ep (w)$, where $\Delta_\ep (w)$ can be computed via the following problem:
    \begin{equation*}
        \Delta_\ep(w)
        :=\min_{\mu\ge 0}
        \PAR{ \tfrac12\norm{ \nabla f(w)+\nabla g_{\mathcal J}(w)\mu }^2_{H^{-1}} - g_{\mathcal J}(w)^\top \mu},
    \end{equation*}
    which remains inexpensive to solve since it is still a nonnegative least-squares problem.
    This reduces to $\Delta(w)$ when $\ep = 0$ because ${\mathcal J}(w;0)= \mathcal I(w)$ and $g_{\mathcal I}(w)=0$. 
    Moreover, enlarging ${\mathcal J}$ (taking larger $\ep > 0$) yields a monotonically tighter certificate.

\section{Proof of Theorem~\ref{thm:cert_nested_U}}\label{app:degradationbound}
    \noindent
    1) Since $\nabla_U F_l(U ;\hat M \hat \Phi,x_0) = 0$ is equivalent to $U = U^\star(\hat M \hat \Phi;x_0)$ (see (\ref{eq:grad})), \PROB{2}($\hat M$) can be rewritten as:
    \begin{align*}
        \min_{\hat \Phi}\;& F_u(U^\star(\hat M \hat \Phi;x_0);x_0) \  
        \text{s.t.}\   G(U^\star(\hat M \hat \Phi;x_0);x_0)\le 0.
    \end{align*}
    This reduces to (\ref{eq:general_opt}) by setting $f,g$ as in (\ref{eq:set_for_P2gap}).
    From Assumption~\ref{ass:upper}, $F_u(\cdot;x_0)$ is strongly convex (let $\bar H\in\PD^{Nm}$ make it $\bar H$-strongly convex) and $G_u(\cdot;x_0)$ is convex.
    Together with the fact that $U^\star(\hat M \hat \Phi;x_0);x_0)$ is affine with respect to $\hat \Phi$, both $f(w),\, g(w)$ become convex.
    In particular, $f(w)$ is $\hat H$-strongly convex with $\hat H = W^\top \bar H W \in \PD^{\dim \hat \Phi}$ because $W:=\bigl( \bar B^\top \bar Q\bar B + \bar R \bigr)^{-1} \bar \Gamma \hat M \in \R^{Nm\times (\dim \hat \Phi)}$ has full column rank.
    Since \PROB{2}($M$) is feasible and $M \Phi^\star = \hat M T \Phi^\star$ holds, $\hat\Phi = T \Phi^\star$ is feasible for \PROB{2}($\hat M$).
    Define $\nabla_{(\cdot)} G^\star(\Theta) := \nabla_{(\cdot)} G_\mathcal{I}( U^\star(\Theta) )$, then MFCQ at $\Phi^\star$ for \PROB{2}($M$) implies that there exists $d$ s.t. 
    $ \nabla_\Phi   G^\star( M\Phi^\star )^\top d \allowbreak
    = \nabla_U      G^\star( M\Phi^\star )^\top \allowbreak
      \nabla_\Theta G^\star( M\Phi^\star )^\top M d M \allowbreak
    = \nabla_U      G^\star( \hat M T \Phi^\star )^\top \allowbreak
      \nabla_\Theta G^\star( \hat M T \Phi^\star )^\top \hat M T d  \allowbreak
    = \nabla_{\hat\Phi} g_\mathcal{I}(T \Phi^\star)^\top (Td)  \allowbreak < 0$, which means MFCQ at $T\Phi^\star$ for \PROB{2}($\hat M$).
    Therefore, we can apply (\ref{eq:gap_general}) with $w = T \Phi^\star$ to prove the claim.\vspace{5pt}\\
    2) \PROB{2} can be rewritten as \eqref{eq:general_opt} by setting $f,g$ as \eqref{eq:set_for_P2gap} with $\hat M = I_{Nm}, \, \hat \Phi = \Theta$.
    Under $G_u(U_0;x_0)\leq 0$, $(\Theta^\star(U_0;x_0),U_0)$ is feasible for \PROB{2}.
    Therefore, applying this $w=\Theta^\star(U_0;x_0)$ to (\ref{eq:gap_general}) leads to $0 \le V_0 - V_2 \le \Delta( \Theta^\star(U_0;x_0) )$.

\section{Quadrotor Control via Koopman EDMD}\label{app:quadrotorexample}
    
    Our method can be applied to a wide range of nonlinear systems by using linearizing techniques, e.g., Koopman extended dynamic mode decomposition (EDMD) \cite{korda2018linear}, Hammerstein--Wiener~\cite{moriyasu2021structured}, etc.
    We address the control of quadrotors as a specific example, referring to \cite{NARAYANAN2023607}.

    \subsection{Modeling and Pre-stabilization}

    We consider the following dynamics for a quadrotor:
    \begin{align*}
        \dot x_p &= x_v, \ \dot x_v = \frac{v_t}{m} x_r e_z - g e_z, \\
        \dot x_R &= x_R \hat \omega ,\ \dot x_\omega = J^{-1} ( v_m - x_\omega \times J x_\omega ),
    \end{align*}
    where $x_p\in\R^3$ is the position, $x_v\in\R^3$ is the velocity, $x_R \in \mathrm{SO}(3)\subset \R^{3\times 3}$ is the rotation matrix, $x_\omega \in \R^3$ is the angular velocity, $m = 4.34$~kg is the mass, $J=\diag{0.0820, 0.0845, 0.1377} \mathrm{~kg\cdot m^2}$ is the matrix of the moment of inertia, $g = 9.81~\mathrm{m/s^2}$ is the gravitational acceleration, $e_z = [0;0;1]$ is the unit vector, and $\hat{(\cdot)}:\R^3 \rightarrow \text{SO(3)}$ is the hat operator: $a \times b = \hat a b $ for $a,b\in \R^3$.
    It has a thrust force $v_t\in\R$ and a total moment $v_m\in\R^3$ as control inputs, both are defined in the body frame.

    An LTI model is learned from the above nonlinear dynamics via Koopman EDMD:
    \begin{align*}
        x_{k+1} &= A_0 x_k + B v_k, \ v = [u_t; \, v_m] \in \R^4, \\
        x &= [x_p; \, x_v;\, \text{vec}(x_R-I);\, x_\omega;\, \\
          & \ \ \ \ (x_R-I)e_z;\,
            \text{vec}(x_R \hat \omega);\, J^{-1} \hat{\omega} (J\omega) ] \in \R^{27},
    \end{align*}
    with time-step $\Delta t = 0.03$~s. 
    Although the learning process is omitted as it is beyond the scope, if the dynamics are learned appropriately, $A_0$ should be unstable and approximately satisfy $(x_p)_{k+1} = (x_p)_{k} + (x_v)_k \Delta t$, i.e., $A_0$ has integrators (eigenvalues 1), and equivalently, $I - A_0$ is singular.
        
    To avoid internal divergence in MPC computation due to instability, we apply a local pre-stabilizing feedback $v_k = -K x_k + u_k$ with a new input $u_k \in \R^4$, resulting in
    \begin{align} \label{eq:semistab}
        x_{k+1} &= A x_k + B u_k, \ A := A_0 - B K. 
    \end{align}
    The gain $K$ is designed to preserve the integrator modes, to prevent the stable position from being fixed by this local control\footnote{Assume $A_0$ includes $l$ integrators ($\rank (I-A_0) = n-l$), let $V\in \R^{n\times l},\,W\in \R^{n\times (n-l)}$ be basis matrices of $\ker (I-A_0)$ and $\im (I-A_0)$, respectively (so $T:=[W,V]$ is nonsingular). This yields $\tilde A_0 := T^{-1} A_0 T = [A_{11},O; A_{21}, I_l],\, \tilde B := T^{-1} B = [B_1;B_2]$, and then we can design $K_1$ to allow $A_{11} - B_1 K_1$ to have eigenvalues within the unit circle. Finally, we have $\tilde K := K T = [K_1, O]$, preserving integrators in $A = T (\tilde A_0 - \tilde B \tilde K) T^{-1} = T [A_{11} - B_1 K_1, O; A_{21} - B_2 K_1, I_l] T^{-1}$ .}.
    
    \subsection{Lower-level MPC}

    We construct the lower-level tracking MPC \PROB{L} (\ref{eq:lower}) for the pre-stabilized LTI (\ref{eq:semistab}), mainly aiming to track given position trajectories.
    The objective function $F_l(U;\Theta,x_0)$ is set as in Example~\ref{ex:example1}, and in this case, $F_l$ can be rewritten as:
    \begin{align*}
        F_l(U;\Theta,x_0) = & \ \frac{1}{2} \norm {x_N - \bar x_N}_P^2 \\
        & + \frac{1}{2} \sum_{k=0}^{N-1} ( \norm{x_k - \bar x_k}_Q^2 + \norm{u_k - \bar u_k}_R^2 ), 
    \end{align*}
    where $[ \bar x_{k+1} ; \bar u_k ] = Z \theta_k \, (k=0,\ldots,N-1), \, \bar x_0 = x_0$.
    We set $N=30$, $Q = \diag{ I_3,\, 10^{-4}I_{24} } \succ 0$, aiming to impose a larger cost on the position error, and $R= 0.01 I_4 \succ 0$.
    The terminal weight $P \succ 0$ is determined by solving the discrete-time Lyapunov equation $(A-BK_f)^\top P (A-BK_f) - P = - (Q + K_f^\top R K_f)$, where $K_f$ is the virtual gain obtained by discrete-time LQR with weights $(Q,R)$.
    Since $I-A$ is singular due to integrators in $A$, the matrix $Z = [Z_x; Z_u]$ is chosen from $\ker S \ (S=[I-A,-B])$; specifically, we take it by \texttt{Z = null(S)} in MATLAB.
    In addition, we impose the stage-wise input box constraint $-20\cdot[2;1;1;1]\le v_k\, (=-Kx_k + u_k) \le 20\cdot[2;1;1;1] \, (k=0,...,N-1)$ encoded as $g_l(x_k,u_k) \leq 0 \ (k=0,...,N-1)$, and by stacking them, we have $G_l(U;x_0) = [g_l(x_0,u_0); g_l(x_1,u_1); ...; g_l(x_{N-1}, u_{N-1}) ]\le 0$.
    The terminal set is omitted here for simplicity, so stability is not fully guaranteed; however, the terminal cost ensures local asymptotic stability around $(x,u)=(0,0)$ in the case $(\bar x_{k+1}, \bar u_k) = (0,0) \ (k=0,...,N-1)$, since that point is a steady-state in the interior of the feasible set.
    The initial state is set to $x_0=0$.
    
    \subsection{Upper-level MPCs}

    As an upper-level goal, we assume a fixed target state $\check x = [\check x_p; 0; 0; ...; 0],\,\check x_p = [0;0;1]$, target input $\check u = 0$, and the $x,z$--position bound $x_{p_x} \leq 0,\, x_{p_z} \leq 0$ encoded as $g_u(x,u) \leq 0$.
    Therefore, the target state is outside the feasible set.
    Under this position constraint and the lower-level input constraint, we aim to minimize the upper-level cost: 
    \[
        \ell(x,u) := \frac{1}{2} \PAR{ \norm{x_k - \check x}_Q^2 + \norm{u_k - \check u}_R^2 }.
    \]
    Here, we used the same $Q$ and $R$ as the lower-level, for simplicity.
    
    Before considering a bilevel MPC, we construct HMPC \PROB{0} \eqref{eq:P0} as a baseline, using $\ell(x,u), g_u(x,u), g_l(x,u)$.
    For our problem setup, the solution is $(x^*,u^*) = (0,0) \ (\theta^*=0)$.

    Next, we construct bilevel MPC \PROB{1} \eqref{eq:P1}, its reduction \PROB{2} \eqref{eq:P2}, their move-blocked variants \PROB{1}($M$) \eqref{eq:P1M} and \PROB{2}($M$) \eqref{eq:P2M}, and the corresponding centralized MPC \PROB{3} \eqref{eq:P3}.
    The objective and inequality constraint functions are set as
    \begin{align*}
        F_u(U;x_0) &:= V_f(x_N) + \sum_{k=0}^{N-1} \ell(x_k, u_k), \\
        V_f(x) &:= \frac{1}{2} \norm{x - x^*}_P^2 - \nu^\top (x-x^*),\\
        G_u(U;x_0) &:= \PPPAR{ \ g_u(x_0,u_0) \ ; \ \cdots \ ;\ g_u(x_{N-1},u_{N-1}) \ },
    \end{align*}
    where $\nu$ is a vector satisfying $(S Z)^\top \nu = 0$\footnote{Let $z:=[x;u], \ g(z):=[g_u(x,u);g_l(x,u)]$, then $\nu$ corresponds to the Lagrange multiplier for the equality constraint in $\mathbf{P}_0^\prime: \min_z \left\{ \ell(z)\, \middle|\, g(z)\le 0, S z = 0 \right\}$, which is equivalent to \PROB{0}. 
    From its KKT condition, we have $\nabla_z \ell(z) + \nabla_z g(z) \mu + S^\top \nu = 0 \Rightarrow Z^\top (\nabla_z \ell(z) + \nabla_z g(z) \mu) + (SZ)^\top \nu = 0$ with the Lagrange multipliers $\mu, \nu$. 
    From the KKT condition of \eqref{eq:P0}, we also have $\nabla_\theta \ell(z(\theta)) + \nabla_\theta g(z(\theta)) \mu = Z^\top (\nabla_z \ell(z) + \nabla_z g(z) \mu) = 0$. 
    Therefore, we have $(SZ)^\top \nu = 0$.}.
    For simplicity, we again omit terminal constraints here. 
    However, since this setting establishes strict dissipativity\footnote{Define the storage function $\lambda(x) := \nu^\top(x-x^*)$, and let $\mu$ be the Lagrange multiplier for $g(z)\leq 0$ at the solution $z = z^* :=[x^*;u^*] $ of $\mathbf{P}_0^\prime$. Due to $\mu\geq 0, \, g(z^*)\leq 0$, and convexity of $g(z)$, $(\nabla_z g(z^*) \mu)^\top (z-z^*)\leq  \mu^\top g(z) \leq 0 \ \forall z \text{ s.t. } g(z)\leq 0$ holds. Considering it with $\nabla_z \ell(z^*) + \nabla_z g (z^*)\mu + S^\top \nu = 0$, the rotated stage cost $\tilde \ell(z) := \ell(z) -\ell(z^*) +\lambda(x) - \lambda(Ax+Bu)$ satisfies $\tilde \ell(z) = 1/2 (\norm{x-x^*}_Q^2 + \norm{u-u^*}_R^2) - (\nabla_z g(z^*) \mu)^\top (z-z^*) \geq 1/2 (\norm{x-x^*}_Q^2 + \norm{u-u^*}_R^2)$. It yields the strict dissipativity: $\lambda(Ax+Bu)-\lambda(x)\leq \tilde \ell(z)-\tilde \ell(z^*)-1/2(\norm{x-x^*}_Q^2 + \norm{u-u^*}_R^2)  \ \forall z \text{ s.t. } g(z)\leq 0$.} with respect to $(x,u) = (x^*,u^*)$, practical convergence to that point can be achieved under additional conditions such as recursive feasibility and the turnpike property \cite{7039809}.
    Although the additional assumptions have not been verified, we confirmed that it converged under broad initial conditions, including an example shown below.
         
    \begin{figure}[t!]
            \centering
            \includegraphics[clip,scale=0.8]{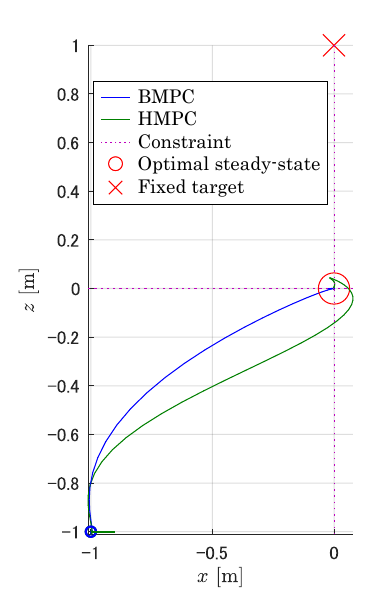}
            \caption{Controlled trajectories with HMPC \PROB{0} (green) and BMPC \PROB{2} (blue) for quadrotor example. Although the fixed target (cross) is infeasible, both controllers converge to the optimal steady state (circle); BMPC does so while satisfying the upper-level constraint (dotted lines), whereas HMPC violates it during the transient.\label{fig:quad_traj}}
            \includegraphics[clip,scale=0.8]{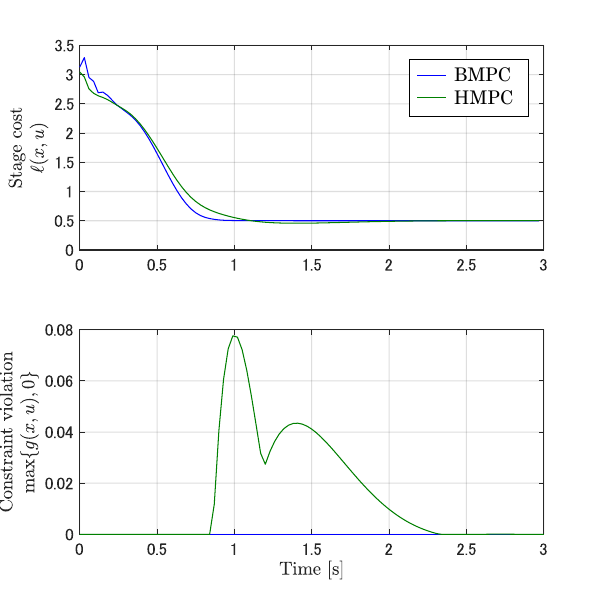}
            \caption{Performance measures of HMPC \PROB{0} (green) and BMPC \PROB{2} (blue) for the quadrotor example. BMPC avoids transient constraint violation and converges faster than HMPC, yielding slightly lower average stage cost.\label{fig:quad_cost}}
    \end{figure}
    
    \subsection{Closed-loop simulation}
    We examined the closed-loop performance of MPCs \PROB{0}--\PROB{3}.
    For two MPCs, HMPC \PROB{0} and the reduced bilevel MPC \PROB{2}, the controlled trajectory and the performance measures are compared in Figs.~\ref{fig:quad_traj} and~\ref{fig:quad_cost}, respectively.
    The results of bilevel MPC \PROB{1} (with zero initialization) and centralized MPC \PROB{3} are omitted, as they were equivalent to the result of \PROB{2}, except for very slight numerical errors.
    In these figures, the blue and green solid lines represent the results of bilevel MPC (BMPC) \PROB{2} and HMPC \PROB{0}, respectively; the purple broken lines show the upper bounds on the $x$ and $z$ positions; and the red circle and cross are the optimal steady-state position and the fixed target, respectively.

    In Fig.~\ref{fig:quad_traj}, we can see that the fixed target is outside the feasible region; the optimal steady position is the closest point to the target in the feasible region; and both BMPC and HMPC converged to the optimal position.
    However, since HMPC handles the upper-level constraint (position bounds) only for the steady-state, it violates the bounds in the transient phase, as seen in Fig.~\ref{fig:quad_cost} (from 0.8~s to 2.3~s).
    In contrast, BMPC satisfied the upper-level constraint in the transient phase.
    In addition, Fig.~\ref{fig:quad_cost} shows that BMPC converged to the point faster than HMPC, as seen in the steeper decrease of the stage cost around 0.7~s.
    The times at which they achieved stabilization around the optimum point without constraint violation are approximately 0.8~s for BMPC and 2.3~s for HMPC.
    The average stage costs were slightly better in BMPC, i.e., 0.892 for BMPC and 0.900 for HMPC, although BMPC must pay extra cost to avoid violating constraints.
    Regarding the upper-level cost function $\ell(x,u)$, HMPC minimizes it only for the steady-state.
    In contrast, BMPC directly optimizes $\ell(x,u)$ as a stage cost in the horizon, yielding the cost advantage in general.

    The average computation times for each MPC are (\PROB{0}) 17.7~ms, (\PROB{1}) 14300~ms, (\PROB{2}) 54.4~ms, and (\PROB{3}) 38.3~ms, where \PROB{0}--\PROB{2} include the computation time for the lower-level MPC.
    This shows that the computational burden for \PROB{2} is about the sum of those for \PROB{0} and \PROB{3}; corresponding to the fact that the size of the upper-level problems in \PROB{0} is significantly smaller compared to the lower-level problems, and the complexity of the upper-level problems in \PROB{2} is nearly equivalent to that of \PROB{3}.
    Since the performance of \PROB{1}--\PROB{3} is equivalent, the most computationally efficient choice is naturally \PROB{3}, when hierarchical control is not required. 
    When hierarchical control is desirable, BMPC \PROB{1} incurs an unrealistic computational cost, necessitating the selection of HMPC \PROB{0} in most past cases.
    The above results demonstrated that using the proposed MPC \PROB{2} allows BMPC to be implemented with a sufficiently fast computation while achieving equivalent performance to \PROB{1} and better reliability than \PROB{1}.

    \subsection{Performance degradation and its bounds}
    Next, we also consider the move-blocked problems \PROB{1}($M$) and \PROB{2}($M$), and verify the results on performance ordering and bounds via single optimization results at the initial condition $x_0=0$ of the above problem.
    We again used move-blocking matrices $M_i$ \eqref{eq:Mi} that make the first $i$ steps in the horizon free and the last $(N-i)$ steps blocked.
    We also prepared a low-rank move-blocking matrix $M^\star$ by generating 50 initial states that have different positions with $x_0$ and giving it to Algorithm~\ref{alg:construct_M_from_P0}.
    The obtained $M^\star$ had $\rank (M^\star)=16$ that was equivalent to $\rank (M_{4})$.
    
    \begin{figure}[t]
            \centering
            \includegraphics[clip,scale=0.8,trim=20 0 0 0]{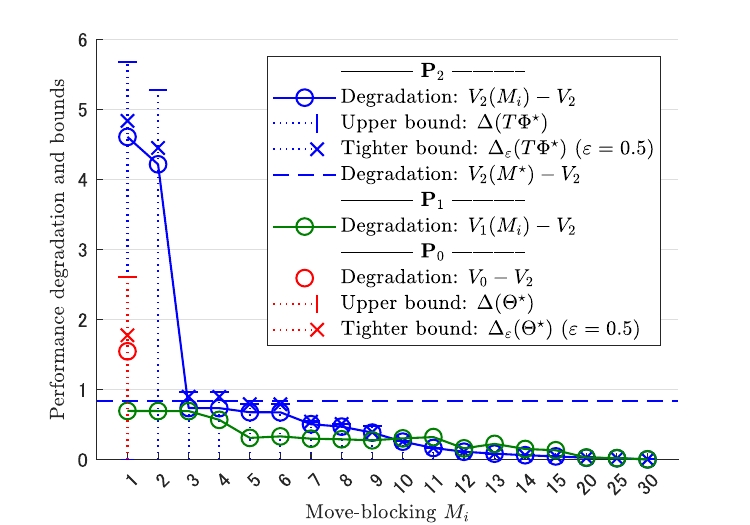}
            \caption{Performance degradation relative to \PROB{2} for move-blocked problems (blue circle: \PROB{2}($M$), green circle: \PROB{1}($M$)) and HMPC \PROB{0} (red circle) in the quadrotor example, together with the proposed upper bounds; the overall trends are similar to those in Fig.~\ref{fig:toyMB}. In addition, the results for \PROB{1}($M$) do not always satisfy the expected monotonicity $V_1(M_{i+1})\leq V_1(M_i)$ or the dominance relation $V_1(M_i)\leq V_2(M_i)$, reflecting the nonconvexity of \PROB{1}($M$) and further motivating the use of \PROB{2}($M$) for reliable control design. \label{fig:MB_and_bounds}}
    \end{figure}  
    Fig.~\ref{fig:MB_and_bounds} shows the degradation ratio of the performance of \PROB{0}, \PROB{1}($M_i$) and \PROB{2}($M_i$) ($i=1,..,30$) compared to \PROB{2}.
    The blue, magenta, and red bullets represent the actual performance degradations $V_2(M_i)-V_2,\, V_1(M_i)-V_2 \ (i=1,...,30)$ and $V_0-V_2$, respectively.
    The bars represent the lower and upper bounds, i.e., the blue and red ones show $0\leq V_2(M_i)-V_2\leq \Delta(T\Phi^\star)$, $0\leq V_0 - V_2\leq \Delta(\Theta^\star)$, respectively.
    The cross markers represent the tighter bounds, i.e., the blue and red ones show $\Delta_\ep (T\Phi^\star)$ and $\Delta_\ep (\Theta^\star)$ with $\ep = 0.5$.
    The green line represents $V_2(M^\star) -V_2$.

    Since the initial condition $x_0=0$ is far from the border of the upper-level position constraint, the result of \PROB{0} satisfies $G_u(U_0;x_0)\leq 0$.
    The figure shows $V_0 \geq V_1(M_1)$ as anticipated from Theorem~\ref{thm:ordering_MB}-3, under $G_u(U_0;x_0)\leq 0$. 
    In contrast, $V_0 \geq V_2(M_1)$ did not hold in this case because some of the lower-level constraints were active; therefore, $\nabla_U F_l(U_0;\Theta_0,x_0)\neq 0$.
    However, the result shows that \PROB{2}($M$) recovers the performance with small refinement of $M$, i.e., $V_0 \geq V_2(M_i) \ (i\geq3)$ and $V_0 \geq V_2(M^\star)$ hold.
    It is notable that $V_2(M_i)-V_2$ decreases drastically with $i= 3$ and the proposed bounds adequately estimate that the performance degradation is small with $i\geq 3$.
    Therefore, when we consider the workflow to decide $M_i$ by increasing $i$ in order, we will find a well-balanced $M_i$ at $i=3$ by checking the proposed bounds.
    The pre-determination of $M^\star$ is also effective in achieving a good balance without iterative computation of \PROB{2}($M_i$).
    The result also shows that the monotonicity $V_2(M_i) \geq V_2(M_{i+1})$ holds for any $i \in [1,N-1]$ as anticipated in Theorem~\ref{thm:ordering_MB}-2, however, $V_1(M_i) \geq V_1(M_{i+1})$ and $V_2(M_i) \geq V_1(M_i)$ did not hold for some $i$, e.g., the figure shows $V_1(M_{10})\leq V_1(M_{11}), \ V_2(M_{11})\leq V_1(M_{11})$, etc.
    These inconsistencies occurred due to the nonconvexity of \PROB{1}($M$), which motivates us to employ \PROB{2}($M$) for reliable control design.   

\end{document}